\documentclass[a4paper,11pt]{article}
\usepackage[english]{babel}
\usepackage{amssymb,amsmath,amsfonts}
\usepackage{hyperref}
\usepackage{amsthm}
\usepackage{color}
\definecolor{red}{rgb}{1,0,0}
\definecolor{noir}{rgb}{0,0,0}

\theoremstyle{theoreme}
\newtheorem{thm}{Theorem}[section]
\newtheorem{lm}[thm]{Lemma}
\newtheorem{prop}[thm]{Proposition}
\newtheorem{cor}[thm]{Corollary}

\newtheorem{hyp}{A\hspace{-.05in}}
\newtheorem{dfn}[thm]{Definition}

\def\N{\mathbb{N}}

\def\R{\mathbb{R}}
\def\C{\mathbb{C}}
\def\P{\mathbb{P}}

\def\Tr{\mathrm{Tr}}
\def\Id{\mathrm{Id}}
\def\Im{\mathrm{Im}}
\def\Real{\mathrm{Re}}
\def\Re{\text{Re}}
\hypersetup{pdfborder={0000}, colorlinks=true, linkcolor=blue}

\begin{document}
\title{Mean field limit for Bosons with compact kernels interactions
  by Wigner measures transportation.}

\author{Boris Pawilowski\footnote{boris.pawilowski@univ-rennes1.fr,
  IRMAR, Universit\'e de Rennes I, campus de Beaulieu, 35042 Rennes
  Cedex, France.} \;, Quentin Liard\footnote{quentin.liard@univ-rennes1.fr, IRMAR, Universit\'e de Rennes I, campus de Beaulieu, 35042 Rennes
  Cedex, France.}}

\maketitle

\begin{abstract}
We consider a class of many-body Hamiltonians composed of a free (kinetic) part and a multi-particle (potential) interaction with a compactness assumption on the latter part.
We investigate the mean field limit of such quantum systems following the Wigner measures
approach. We prove the propagation of these  measures along the flow of
a nonlinear (Hartree) field equation. This enhances and complements some
previous results of the same type shown in \cite{AmNi2,AmNi3,FGS}.
\end{abstract}

\section{Introduction}

In the present paper, we consider the mean field problem for a
system of many quantum particles described by a $N$-body Schr\"odinger  Hamiltonian which
is typically a sum of a kinetic energy and a multi-particle interaction.
It is well-known that the mean field theory  provides a reliable approximation
of the many-body Schr\"odinger dynamics by a one particle (nonlinear) dynamics.
Such approximation is fundamental in physics in modelling of
Bose-Einstein condensates and superfluidity.
Mathematically, it is an interesting question with some subtlety.

There are several approaches to the derivation of the mean field limit.  For instance the method of coherent states \cite{GiVe1,Hep,RoSc} or  the BBGKY hierarchy of reduced density matrices (\cite{BGM,ChPa,ErYa,ESY1,KlMa, KnPi}), to mention only few. There is a newer approach
inspired by ideas from the analysis of oscillation phenomena in finite dimension and relying  on Wigner measures  (see Definition \ref{de.wigmeas} and \cite{AmNi1, AmNi2, AmNi3, AmNi4}).
With this method the propagation of fairly general quantum states can
be understood  in the mean field limit (\cite{AmNi2,AmNi3}). These
quantum states are constrained to a certain compactness property,
called the (PI) condition \eqref{PI},  essential for the
convergence. One expects that loss of compactness would be in fact an obstacle to the achievement of the limit. In \cite{AmNi3} the case of a bounded interaction was proven with the specific (PI) condition  \eqref{PI}.

Our main purpose is to show that under a compactness assumption  on
the interaction, one can describe the propagation of a wider class of
quantum states. The only assumption needed now is an uniform trace estimate reflecting a finite density of particles on these states but not the  (PI) condition \eqref{PI}. 

The main tool used here is the Wigner measure introduced
in \cite{AmNi1} in an infinite dimensional setting. These measures
reflect the ideas of space-phase analysis. It was developed, in finite
dimension, in the work of P.~G\'erard, P.A~Markowich, N.~Mauser, F.~Poupaud
\cite{MMGP}, B.~Helffer , A.~Martinez and D.~Robert \cite{HMR}, P.L Lions and T.~Paul \cite{LiPa} and L.Tartar \cite{Tar}. Our strategy is based on the work of Z.~Ammari and F.~Nier
\cite{AmNi4} and L. Ambrosio, N. Gigli and G. Savare in \cite{AGS} and it requires the study of continuity equations in infinite dimensional spaces.

\bigskip
\bigskip
\noindent
We work in the bosonic Fock space
$$
\Gamma_{s}(\mathcal{Z})=\bigoplus_{n=0}^{\infty}\bigvee^{n}
\mathcal{Z}=\bigoplus_{n=0}^{\infty}\mathcal{S}_{n}\mathcal{Z}^{\otimes n}\,,
$$
modelled on a one particle separable complex Hilbert space
$\mathcal{Z}$\,, where $\mathcal{S}_{n}$ is the symmetrization
projection defined on $\mathcal{Z}^{\otimes{n}}$ by
$$
\mathcal{S}_{n}(\varphi_{1}\otimes \cdots\otimes
\varphi_{n})=\frac{1}{n!}\sum_{\sigma\in
  \varSigma_{n}}\varphi_{\sigma(1)}\otimes\cdots\otimes \varphi_{\sigma(n)}\,,
$$ with the sum running over all permutations of $n$ elements. If not specified tensor products and orthogonal direct sums are
considered in their  Hilbert completed version.
We are interested in the mean field dynamics of
the many-body Hamiltonian with multi-particle interaction
\begin{equation}
\label{eq.Hneps}
H^{(n)}_{\varepsilon}=H_{\varepsilon}^{0,(n)}
+
\sum_{\ell=2}^{r}\varepsilon^{\ell}\frac{n!}{(n-\ell)!}\mathcal{S}_{n}(\tilde{Q}_{\ell}\otimes
\mathrm{Id}_{\bigvee^{n-\ell}\mathcal{Z}})\mathcal{S}_{n}\,,\quad n\geq
2r\,,
\end{equation}
in the asymptotic regime $\varepsilon\to 0$\,, $n\varepsilon\to 1$. Here
the $\tilde Q_\ell$'s are bounded symmetric operators on $\bigvee^{\ell}\mathcal{Z}$ and
\begin{equation}
  \label{eq.H0neps}
H_{\varepsilon}^{0,(n)}=\varepsilon
\sum_{i=1}^{n}\mathrm{Id}\otimes\cdots\otimes \mathrm{Id}\otimes
\underbrace{A}_{i}\otimes\mathrm{Id}\otimes \cdots\otimes \mathrm{Id}\,,
\end{equation}
where $A$ is a given self-adjoint operator.

Within the second quantization (see Appendix~\ref{se.wickq}), the operator
$H_{\varepsilon}^{(n)}$ (resp. $H_{\varepsilon}^{0,(n)}$) can be
written as a restriction to the subspace $\bigvee^{n}\mathcal Z$ of
the operator $H_{\varepsilon}$ (resp. $H_{\varepsilon}^{0}$)
defined on the Fock space and given by:
\begin{eqnarray}
\label{eq.HepsWick}
  && H_{\varepsilon}=H_{\varepsilon}^{0}+ Q^{Wick}\quad,\quad
  Q(z)=\sum_{\ell=2}^{r}\langle
z^{\otimes \ell}\,,\, \tilde{Q}_{\ell}z^{\otimes
  \ell}\rangle,\\
\label{eq.Heps0Wick}
&& H_{\varepsilon}^{0}=\mathrm{d}\Gamma(A)=\langle z\,,\, A z\rangle^{Wick}\,.
\end{eqnarray}
The mean field energy functional is
\begin{equation}
h(z,\bar z)=\langle z,Az\rangle + \; Q(z)\,,
\end{equation}
so that the mean field dynamics are given by the non linear equation

\begin{equation}
\label{eq.hartree0}
i \partial_{t} z_{t} = \partial_{\bar z}h(z_{t},\bar z_{t})=Az_{t}+\partial_{\bar z}Q(z_{t}).
\end{equation}
In our framework, the annihilation and creation operators,
$a(z_{1})$ and $a^{*}(z_{2})$\,, with $z_{1},z_{2}$ in $\mathcal{Z}$\,,
 satisfy the $\varepsilon$-dependent Canonical Commutation
Relations (CCR):
$$
\left[a(z_{1}), a^{*}(z_{2})\right]=\varepsilon \langle z_{1}\,,\,
z_{2}\rangle \;\Id\,.
$$
Recall that the Weyl operator, for $\xi\in \mathcal{Z}$, is defined by
$$
W(\xi)=e^{i\frac{a(\xi)+a^{*}(\xi)}{\sqrt{2}}},
$$
and the number operator $\mathbf{N}$ is
$$
\mathbf{N}=\mathrm{d}\Gamma(\Id)\,.
$$
We refer the reader to Appendix~\ref{se.main} for a brief review of
the second quantization and these related operators.

Our approach in the derivation of the mean field dynamics uses  Wigner
measures. For reader convenience, we recall the definition below.
\begin{dfn}
\label{de.wigmeas}
Let $\mathcal{E}$ be an infinite subset of
$(0,+\infty)$ such that $0\in \overline{\mathcal{E}}$\,.
Let $(\varrho_{\varepsilon})_{\varepsilon\in
  \mathcal{E}}$ be a family of normal states on
$\Gamma_{s}(\mathcal{Z})$\,
($\varrho_{\varepsilon}\geq 0$ and
$\Tr\left[\varrho_{\varepsilon}\right]=1$) such that:
$$
\exists \delta>0\,, \exists C_{\delta}>0\,,  \forall \varepsilon\in
\mathcal{E}\,,\quad
\Tr[\varrho_{\varepsilon}\mathbf{N}^{\delta}] \leq C_{\delta}<\infty\,.
$$
The set $\mathcal{M}(\varrho_{\varepsilon}, \varepsilon \in
\mathcal{E})$ of Wigner measures associated with $(\varrho_{\varepsilon})_{\varepsilon\in
\mathcal{E}}$ is the set of Borel probability
measures on $\mathcal{Z}$\,, $\mu$\,, such that there exists an
infinite subset $\mathcal{E}'\subset \mathcal{E}$ with
$0\in \overline{\mathcal{E}'}$ and
$$
\forall \xi\in \mathcal{Z}\,,\quad
\lim_{\mathcal{E}'\ni
    \varepsilon \to 0
}
\Tr\left[\varrho_{\varepsilon}W(\sqrt{2}\pi\xi)\right]=
\int_{\mathcal{Z}}e^{2i\pi\Real\langle \xi,z\rangle}d\mu(z)\,.
$$
\end{dfn}
By some diagonal extraction of subsequences,
it was proved in \cite{AmNi1} that $\mathcal{M}(\varrho_{\varepsilon},
\varepsilon\in \mathcal{E})$ is never empty
 (see Theorem~\ref{th.Wigdef}).

Our main result will be proved under the following assumptions:\\
\begin{hyp}
\label{hyp.A}
The operator $A$ with the domain $D(A)$ is self-adjoint in $\mathcal{Z}$\,.
\end{hyp}
\begin{hyp}
\label{hyp.Q}
For all $\ell\in \left\{2,\ldots,r\right\}$\,, the operator
$\tilde{Q}_{\ell}$ is compact and self-adjoint in
$\bigvee^{\ell}\mathcal{Z}$\,.\\
\end{hyp}
We will prove the following result.
\begin{thm}
\label{thm.main}
Let $(\varrho_{\varepsilon})_{\varepsilon \in (0,\bar{\varepsilon})}$ be
a family of normal states on $\Gamma_{s}(\mathcal Z)$
such that
\begin{equation}
\label{eq.di}
\exists \delta>0\,, \exists C_{\delta}>0\,,  \forall \varepsilon\in
(0,\bar\varepsilon)\,,\quad
\Tr[\varrho_{\varepsilon}\mathbf{N}^{\delta}] \leq C_{\delta}<\infty \;,
\end{equation}
and which admits a unique  Wigner measure $\mu_{0}$\,.
Under the Assumptions {\bf{(A1)}}-{\bf{(A2)}}, the family
$(e^{-i\frac{t}{\varepsilon}H_{\varepsilon}}\varrho_{\varepsilon}e^{i\frac{t}{\varepsilon}H_{\varepsilon}})_{\varepsilon
  \in (0,\bar{\varepsilon})}$ admits for every $t\in \R$
a unique Wigner measure $\mu_{t}$\,, which is the push-forward
$\Phi(t,0)_{*}\mu_{0}$ of the initial measure $\mu_{0}$ by the
 flow associated with
\begin{equation}
\label{eq.Cauchy}
\left\{
      \begin{aligned}
      i  \partial_{t}z_{t}&=Az_{t}+\partial_{\bar z}Q(z_{t}),&\\
        z_{t=0}&=z_{0}.\\
         \end{aligned}
    \right.
\end{equation}
\end{thm}
As we have mentioned  previously, many authors use the BBGKY hierarchy
method to justify the mean field limit. Here the analysis is different
but it is still possible to formulate our results by using the vocabulary
of reduced density matrices. For a family of normal states
$(\varrho_{\varepsilon})_{\varepsilon \in (0,\bar{\varepsilon})}$ on
$\Gamma_{s}(\mathcal Z)$ and $p \in \N$, the reduced density
matrices $\gamma_{\varepsilon}^{(p)}$ such that
$\Tr[\gamma_{\varepsilon}^{(p)}]<+\infty$ are defined according
to $$\Tr[\gamma_{\varepsilon}^{(p)} \tilde b]=\frac{1}{\Tr[\varrho_{\varepsilon}(|z|^{2p})^{Wick}]}\Tr[\varrho_{\varepsilon}b^{Wick}],
\; \forall \tilde b \in \mathcal{L}(\bigvee^{p}\mathcal Z),$$
with the convention that the right-hand side is $0$ when
$$\Tr[\varrho_{\varepsilon}(|z|^{2p})^{Wick}]=0 \;,$$ and $p>0$ (see
\cite{AmNi3} for more details).
The notations $b^{Wick}$ and $\mathcal P_{p,q} (\mathcal Z)$ are
explained in Appendix \ref{se.wickq}. 
\begin{thm}
Let $(\varrho_{\varepsilon})_{\varepsilon \in (0,\bar{\varepsilon})}$
be a family of normal states on $\Gamma_{s}(\mathcal Z)$ with a single Wigner measure
$\mu_{0}$ and satisfying the (PI) condition, i.e.:
\begin{equation}
\label{PI}
\lim_{\varepsilon\to 0} \Tr[\varrho_\varepsilon \mathbf{N}^{p}]=\int_{\mathcal
  Z} |z|^{2p} \, d\mu_{0}(z)<\infty\,,\quad \forall p\in\mathbb{N}\,.
\end{equation}
 Then for all $t \in \R$ and $\varrho_{\varepsilon}(t)=e^{-i\frac{t}{\varepsilon}H_{\varepsilon}}
\varrho_{\varepsilon}e^{i\frac{t}{\varepsilon}H_{\varepsilon}}$,
\begin{equation}
\label{eq.bbgky}
\lim_{\varepsilon \to 0}\Tr[\varrho_\varepsilon(t) b^{Wick}]=\int_{\mathcal
  Z}b(\Phi(t,0)z)d\mu_{0}(z)=\int_{\mathcal Z}b(z)d\mu_{t}(z)\,,
\end{equation}
for any $b \in \mathcal{P}_{alg}(\mathcal Z)=\bigoplus_{p,q \in
  \N}\mathcal{P}_{p,q}(\mathcal Z)$, with $\mu_{t}=\Phi(t,0)*\mu_{0}$
and $\Phi(t,0)$ the flow associated with the equation \eqref{eq.hartree0}. Finally the
convergence of the reduced density matrices $\gamma_{\varepsilon}^{(p)}(t)$
associated with $\varrho_{\varepsilon}(t)$ holds in
$\mathcal{L}^{1}(\bigvee^{p}\mathcal Z),$ the space of trace class operators, with
$$\lim_{\varepsilon \to
  0}\left\|\gamma_{\varepsilon}^{(p)}(t)-\frac{1}{\int_{\mathcal
    Z}|z|^{2p}d\mu_{0}(z)}\int_{\mathcal Z}|z^{{\otimes p}}\rangle
\langle z^{{\otimes p}}|
d\mu_{t}(z)\right\|_{\mathcal{L}^{1}}=0,$$
for all $p \in \N.$
\end{thm}
\begin{proof}
According to Theorem \ref{thm.main}, the family of normal states
$(\varrho_{\varepsilon}(t))_{\varepsilon \in (0,\bar{\varepsilon})}$
admits a single Wigner measure $\mu_{t}$ equal to
$\Phi(t,0)*\mu_{0}$. Besides owing to
$[H_{\varepsilon},\mathbf{N}]=H_{\varepsilon}\mathbf{N}-\mathbf{N}H_{\varepsilon}=0$
the family $(\varrho_{\varepsilon}(t))_{\varepsilon \in
  (0,\bar{\varepsilon})}$ is satisfying the (PI) condition \eqref{PI}
for any time $t \in \R.$ Then Propositions 2.11 and 2.12 in
\cite{AmNi3} allow to obtain the claimed results.

\end{proof}
The proof of Theorem \ref{thm.main} requires few steps.
The operator $H_{\varepsilon}$ with a suitable domain is proved to
be self-adjoint in Proposition~\ref{pr.selfAdj}.
Proposition~\ref{pr.flow} ensures that the Cauchy problem
\eqref{eq.Cauchy} defines a global flow on $\mathcal{Z}$\,.\\
Beside this, the proof consists in several steps that we briefly
sketch here. For the first four points of this proof below we consider
more regular states $\varrho_{\varepsilon}$. \\
\begin{enumerate}
\item By setting
  \begin{eqnarray}
\label{eq.defrho}
&&\varrho_{\varepsilon}(t)=e^{-i\frac{t}{\varepsilon}H_{\varepsilon}}\varrho_{\varepsilon}e^{i\frac{t}{\varepsilon}H_{\varepsilon}}\,,\\
\label{eq.deftilderho}
\text{and}
&&
\tilde \varrho_{\varepsilon}(t)=e^{i\frac{t}{\varepsilon}H_{\varepsilon}^{0}}\varrho_{\varepsilon}(t)e^{-i\frac{t}{\varepsilon}H_{\varepsilon}^{0}}\,,
\end{eqnarray}
we write
\begin{multline}
\label{eq.ipp0}
\Tr[\tilde{\varrho_{\varepsilon}}(t) W(\sqrt{2}\pi\xi)] =
\Tr[\varrho_{\varepsilon}W(\sqrt{2}\pi\xi)]
\\
+ i\int_{0}^{t}{\Tr[\tilde{\varrho_{\varepsilon}}(s)W(\sqrt{2}\pi\xi)
\sum_{j=1}^{r}\varepsilon^{j-1}\mathcal{O}_{j}(s,\xi)] ds}\,,
\end{multline}
where the $\mathcal{O}_{j}(s,\xi)$'s are Wick quantized observables
which satisfy some uniform estimates.
\item The number estimates given in Proposition \ref{number}
 provide equicontinuity   properties of the quantity
  $\Tr\left[\tilde{\varrho}_{\varepsilon}(t)W(\sqrt{2}\pi\xi)\right]$
  w.r.t $(\xi,t)\in \mathcal{Z}\times \R$\,. So that a subsequence
  $(\varepsilon_{k})_{k\in\N}$ converging to $0$ can be extracted such
  that for all times $t\in \R$\,,
$$
\mathcal{M}(\tilde{\varrho}_{\varepsilon}(t), \varepsilon\in \mathcal{E})=\left\{\tilde{\mu}_{t}\right\}
$$
with $\mathcal{E}=\left\{\varepsilon_{k}\,, k\in \N\right\}$\,.
\item With the number estimates, we get rid of the terms for $j\geq 2$ as $\varepsilon\to 0$ in
\eqref{eq.ipp0}\,.
The compactness  assumption {\bf{(A2)}} is used in
Proposition~\ref{pr.compact}
when we take the limit in all
  the remaining terms of \eqref{eq.ipp0} for general initial data
  $\tilde{\varrho}_{\varepsilon}$\,. Subsequently, the measure
  $\tilde{\mu}_{t}$ is a weak solution of the Liouville equation
\begin{equation}
\label{eq.Liouville}
i \partial_{t} \tilde{\mu}_{t} + \{ Q_{t}(z),\tilde{\mu}_{t} \} =0\,,
\end{equation}
with $Q_{t}(z)=Q(e^{-itA}z)$\,.
\item Finally, we follow the same lines as in \cite{AmNi4} and refer
 to measure transportation
  tools developed in \cite{AGS}
in order to prove
$(e^{-itA})_{*}\tilde{\mu}_{t}=\Phi(t,0)_{*}\mu_{0}$\,, and hence we
get $$\mathcal{M}(\varrho_{\varepsilon}(t),\varepsilon \in
(0,\bar{\varepsilon}))=\{\mu_{t}\}.$$
\item
The last point of this proof is a truncation scheme used in
\cite{AmNi4} for more general states.
\end{enumerate}

\section{Quantum and mean-field dynamics}
\label{se.prelim}
In this section we show that the quantum and the classical dynamics are both well defined for all times.
\subsection{Self-adjoint realization}
\begin{prop}
\label{pr.selfAdj}\ \\
(i)  For any $ n \in \N$, the operator $H_{\varepsilon}^{(n)}$ given by \eqref{eq.Hneps} with
 domain $D(\mathrm{d}\Gamma(A)) \cap \bigvee^{n}\mathcal Z$ is a self-adjoint operator in $\bigvee^{n}\mathcal Z$.\\
(ii) The operator $H_{\varepsilon}$, given by \eqref{eq.HepsWick}, is self-adjoint in $ \Gamma_{s}(\mathcal Z)$ with the domain defined by
$$
\left(\Psi\in
D(H_{\varepsilon})\right)
\Leftrightarrow
\left(
\begin{array}[c]{l}
\Psi \in \Gamma_{s}(\mathcal Z), \\
  \forall n \in \N\,,\;\Psi^{(n)} \in
D(H_{\varepsilon}^{(n)})\,,\\
\sum_{n=0}^{\infty} \|H_{\varepsilon}^{(n)}\Psi^{(n)}\|^{2}<+\infty
\end{array}
 \right)\,.
$$
\end{prop}

\begin{proof}
\noindent (i)
For $n\in\N$ and according to \eqref{eq.HepsWick}-\eqref{eq.Heps0Wick}
the operator $H_{\varepsilon}^{(n)}$ equals
\begin{equation}
H_{\varepsilon}^{(n)}=H_{\varepsilon}^{0,(n)}+V_{\varepsilon}^{(n)},
\end{equation}
with  $V_{\varepsilon}^{(n)}=\sum_{\ell=2}^{r}\{ Q_{\ell}(z)
\}_{|\vee^{n}\mathcal Z}^{Wick}$ with $Q_{\ell}(z)=\langle z^{{\otimes
    \ell}},\tilde{Q}_{\ell}z^{{\otimes \ell}}\rangle.$\\
For $\phi^{(n)} \in D(\mathrm{d}\Gamma(A))\cap \bigvee^{n}\mathcal{Z}$\,, a simple computation gives
\begin{align*}
\|V_{\varepsilon}^{(n)}\phi^{(n) }\| &\leq
\sum_{\ell=2}^{r}\varepsilon^{\ell}\frac{n!}{(n-\ell)!}
\|\mathcal{S}_{n}(\tilde{Q_{\ell}}\otimes\Id_{\bigvee^{n-\ell}\mathcal{Z}})\phi^{(n) }\|, &\\
&\leq
\sum_{\ell=2}^{r}\varepsilon^{\ell}\frac{n!}{(n-\ell)!}\|\tilde{Q}_{\ell}\|
\;
 \|\phi^{(n)}\|_{\bigvee^{n}\mathcal Z},&\\
&\leq C_{r,\varepsilon,n} \|\phi^{(n)} \|_{\bigvee^{n}\mathcal{Z}}.\\
\end{align*}
So $V_{\varepsilon}^{(n)}$ is a bounded self-adjoint perturbation of
$H_{\varepsilon}^{0,(n)}$ and therefore $H_{\varepsilon}^{(n)}$ is self-adjoint on $D(H_{\varepsilon}^{0,(n)})$.
\\
\noindent (ii) Proposition A.1 in \cite{AmNi4} is applied here, with $A_{n}=H^{0,(n)}_{\varepsilon}+V_{\varepsilon}^{(n)}$
yields the self-adjointness of $H_{\varepsilon}$\,.
\end{proof}

\noindent
Once we have defined the quantum dynamics, we can then write an integral formula giving the
propagation of normal states. However, instead of considering
\begin{equation*}
\varrho_{\varepsilon}(t)=e^{-i\frac{t}{\varepsilon}H_{\varepsilon}}\varrho_{\varepsilon}e^{i\frac{t}{\varepsilon}H_{\varepsilon}},
\end{equation*}
we will rather work with
\begin{equation*}
\tilde{\varrho_{\varepsilon}}(t)=e^{i\frac{t}{\varepsilon}H_{\varepsilon}^{0}}\varrho_{\varepsilon}(t)e^{-i\frac{t}{\varepsilon}H_{\varepsilon}^{0}}.
\end{equation*}

With the convention of Appendix \ref{se.wickq}, $\mathrm{D}^{j}b$ denotes the j-th differential of $b$ with respect to
$(z,\bar{z})$:
\begin{equation*}
\mathrm{D}^{j}[b(z)] [\xi]=\sum_{|\alpha|+|\beta|=j}\frac{j!}{\alpha!\beta!}\langle
\xi^{\otimes \beta}\; , \;\partial_{z}^{\alpha}\partial_{\bar
z}^{\beta}b(z){\xi}^{\otimes \alpha}\rangle.
\end{equation*}

\begin{prop}
Let $(\varrho_{\varepsilon})_{\varepsilon \in (0,\bar \varepsilon)}$ be a
family of normal states on $\Gamma_{s}(\mathcal Z)$. Assume that
\begin{equation}
\forall k \in \N, \;\exists  C_k>0 \; ,\forall \varepsilon \in (0,\bar \varepsilon), \; \Tr[\varrho_{\varepsilon}\mathbf{N}^{k}] \leq C_{k}.
\end{equation}
Then for all $\xi \in \mathcal Z$,\; the function \;$t \mapsto
\Tr[\tilde{\varrho_{\varepsilon}}(t)W(\sqrt{2}\pi\xi)]$ belongs to $C^{1}(\R)$ and the following formula holds:
\begin{equation}
\label{eq.ipp}
\Tr[\tilde{\varrho_{\varepsilon}}(t)W(\sqrt{2}\pi\xi)] =
\Tr[\varrho_{\varepsilon}W(\sqrt{2}\pi\xi)] +\\
\end{equation}
\begin{equation*}
 i\int_{0}^{t}{\Tr[\tilde{\varrho_{\varepsilon}}(s)W(\sqrt{2}\pi\xi) \{ \sum_{j=1}^{r}\varepsilon^{j-1}\frac{(i\pi)^{j}}{j!}\mathrm{D}^{j}[Q(e^{-isA}z)][\xi]\}^{Wick}] ds}.
\end{equation*}
\end{prop}
\begin{proof}
We denote $\langle \mathbf N \rangle=(1+\mathbf{N}^{2})^{\frac{1}{2}}$. \\
The quantity
$\Tr[(\tilde{\varrho_{\varepsilon}}(t)-\tilde{\varrho_{\varepsilon}}(s))W(\sqrt{2}\pi\xi)]$
 is actually equal to
\begin{equation}
\label{eq.l1}
\Tr[\varrho_{\varepsilon}\langle\mathbf N \rangle^{r} (e^{i\frac{t}{\varepsilon}H_{\varepsilon}}e^{-i\frac{t}{\varepsilon}H_{\varepsilon}^{0}}-e^{i\frac{s}{\varepsilon}H_{\varepsilon}}e^{-i\frac{s}{\varepsilon}H_{\varepsilon}^{0}})
\langle \mathbf N \rangle ^{-r}W(\sqrt{2}\pi\xi)
e^{i\frac{t}{\varepsilon}H_{\varepsilon}^{0}}e^{-i\frac{t}{\varepsilon}H_{\varepsilon}}]
\end{equation}
\begin{equation}
\label{eq.l2}
+\Tr[\varrho_{\varepsilon}\langle \mathbf N
\rangle^{r}e^{i\frac{s}{\varepsilon}H_{\varepsilon}}e^{-i\frac{s}{\varepsilon}H_{\varepsilon}^{0}}\langle
\mathbf N\rangle^{-r}W(\sqrt{2}\pi\xi)\langle\mathbf N\rangle^{r}
(e^{i\frac{t}{\varepsilon}H_{\varepsilon}^{0}}e^{-i\frac{t}{\varepsilon}H_{\varepsilon}}-e^{i\frac{s}{\varepsilon}H_{\varepsilon}^{0}}e^{-i\frac{s}{\varepsilon}H_{\varepsilon}})\langle\mathbf
N \rangle^{-r}]. \\
\end{equation}
By differentiating first for $u \in D(H_{\varepsilon}^{0}) \cap
\bigvee^{n}\mathcal Z$ and then extending the result by continuity, we
get for $u \in \bigvee^{n}\mathcal Z$:
\begin{equation*}
e^{i\frac{s}{\varepsilon}H_{\varepsilon}}e^{-i\frac{s}{\varepsilon}H_{\varepsilon}^{0}}u=
e^{i\frac{t}{\varepsilon}H_{\varepsilon}}e^{-i\frac{t}{\varepsilon}H_{\varepsilon}^{0}}u+
\frac{i}{\varepsilon}\int_{t}^{s}e^{i\frac{\sigma}{\varepsilon}H_{\varepsilon}}
Q^{Wick}e^{-i\frac{\sigma}{\varepsilon}H_{\varepsilon}^{0}}u
\;d\sigma\,.
\end{equation*}
The number estimate in Proposition \ref{number} combined with
$e^{i\frac{s}{\varepsilon}H_{\varepsilon}}e^{-i\frac{s}{\varepsilon}H_{\varepsilon}^{0}}=
\bigoplus_{n=0}^{\infty}e^{i\frac{s}{\varepsilon}H_{\varepsilon}^{(n)}}e^{-i\frac{s}{\varepsilon}H_{\varepsilon}^{0,(n)}},$
implies for all $u \in \Gamma_{s}(\mathcal Z)$:
\begin{equation*}
\langle \mathbf N \rangle
^{-r}e^{i\frac{s}{\varepsilon}H_{\varepsilon}}e^{-i\frac{s}{\varepsilon}H_{\varepsilon}^{0}}u=\langle
\mathbf N \rangle
^{-r}e^{i\frac{t}{\varepsilon}H_{\varepsilon}}e^{-i\frac{t}{\varepsilon}H_{\varepsilon}^{0}}u+
\frac{i}{\varepsilon}\int_{t}^{s}e^{i\frac{\sigma}{\varepsilon}H_{\varepsilon}}\langle
\mathbf N \rangle
^{-r}Q^{Wick}e^{-i\frac{\sigma}{\varepsilon}H_{\varepsilon}^{0}}u \;d\sigma\,.
\end{equation*}
The integrand is continuous in
$\Gamma_{s}(\mathcal Z)$ w.r.t $\sigma$ for any $u \in
\Gamma_{s}(\mathcal Z)$. Taking the limit as $s \to t$ leads to
\begin{equation}
\label{eq.toto8}
s-\lim\limits_{s \to t} \frac{1}{t-s}\langle\mathbf N \rangle^{-r}(e^{i\frac{t}{\varepsilon}H_{\varepsilon}}e^{-i\frac{t}{\varepsilon}H_{\varepsilon}^{0}}-e^{i\frac{s}{\varepsilon}H_{\varepsilon}}e^{-i\frac{s}{\varepsilon}H_{\varepsilon}^{0}})= \\
\frac{i}{\varepsilon}\langle\mathbf N\rangle^{-r}e^{i\frac{t}{\varepsilon}H_{\varepsilon}}Q^{Wick}e^{-i\frac{t}{\varepsilon}H_{\varepsilon}^{0}}.
\end{equation}
Similarly (by exchanging $H_{\varepsilon}^{0}$ and $H_{\varepsilon}$)
we get:
\begin{equation}
\label{eq.toto9}
s-\lim\limits_{s \to t}  \frac{1}{t-s}\langle\mathbf N\rangle^{-r}(e^{i\frac{t}{\varepsilon}H_{\varepsilon}^{0}}e^{-i\frac{t}{\varepsilon}H_{\varepsilon}}-e^{i\frac{s}{\varepsilon}H_{\varepsilon}^{0}}e^{-i\frac{s}{\varepsilon}H_{\varepsilon}}) =-\frac{i}{\varepsilon}\langle\mathbf N\rangle^{-r}e^{i\frac{t}{\varepsilon}H_{\varepsilon}^{0}}Q^{Wick}e^{-i\frac{t}{\varepsilon}H_{\varepsilon}}.
\end{equation}
Notice that
\begin{equation*}\langle\mathbf N \rangle^{-r}e^{i\frac{t}{\varepsilon}H_{\varepsilon}}Q^{Wick}e^{-i\frac{t}{\varepsilon}H_{\varepsilon}^{0}}
 \in \mathcal{L}(\Gamma_{s}(\mathcal Z)),
\quad
\Tr[\varrho_{\varepsilon}\langle\mathbf N \rangle^{r}] < C_{r} <+\infty,
\end{equation*} and
\begin{equation*}
W(\sqrt{2}\pi\xi)e^{i\frac{t}{\varepsilon}H_{\varepsilon}^{0}}e^{-i\frac{t}{\varepsilon}H_{\varepsilon}}
\in \mathcal{L}(\Gamma_{s}(\mathcal Z)).
\end{equation*}
 Thus the trace \eqref{eq.l1} divided by $t-s$ is well defined and converges as $s
 \to t$ thanks to  \eqref{eq.toto8}.
In equation \eqref{eq.toto9}, remark that
\begin{equation*}
\langle \mathbf N \rangle^{-r}e^{i\frac{t}{\varepsilon}H_{\varepsilon}^{0}}Q^{Wick}e^{-i\frac{t}{\varepsilon}H_{\varepsilon}}
  \in \mathcal{L}(\Gamma_{s}(\mathcal Z)),
\end{equation*}
and
\begin{equation*}
\forall \xi \in \mathcal Z, \; \langle \mathbf N \rangle
^{-r}W(\sqrt{2}\pi\xi)\langle \mathbf{N}\rangle^{r} \,
\in \mathcal{L}(\Gamma_{s}(\mathcal Z)),
\end{equation*}
owing to the Lemma 6.2 in \cite{AmBr}. Since for all $u \in \Gamma_{s}(\mathcal Z) $
\begin{equation*}
s \mapsto
e^{i\frac{s}{\varepsilon}H_{\varepsilon}}e^{-i\frac{s}{\varepsilon}H_{\varepsilon}^{0}}u
\in
\mathcal{C}(\R,\Gamma_{s}(\mathcal Z)),
\end{equation*}
the trace \eqref{eq.l2} divided by $t-s$ is well defined
and converges as $s \to t$ thanks to \eqref{eq.toto9}. Therefore the following integral formula holds true, with the help of \eqref{eq.toto7},
\begin{multline}
\Tr[\tilde{\varrho_{\varepsilon}}(t)W(\sqrt{2}\pi\xi)] =
\Tr[\varrho_{\varepsilon}W(\sqrt{2}\pi\xi)] + \\
\frac{i}{\varepsilon}\int_{0}^{t}{\Tr[\tilde{\varrho_{\varepsilon}}(s)[Q_{s}^{Wick}W(\sqrt{2}\pi\xi)-W(\sqrt{2}\pi\xi)Q_{s}^{Wick}]] ds},
\end{multline}
with $Q_{s}(z)=Q(e^{-isA}z)$. We conclude by using \eqref{eq.toto10} and
\eqref{eq.toto11} in Appendix B.
\end{proof}

\subsection{The nonlinear (Hartree) equation}
In this section we shall prove the global well posedness of the mean field dynamics. So we consider the Cauchy problem in $\mathcal Z$:
\begin{equation}
\label{eq.hartree}
\left\{
      \begin{aligned}
      i  \partial_{t}z_{t}&=Az_{t}+\partial_{\bar z}Q(z_{t}),&\\
        z_{t=0}&=z_{0}.\\
         \end{aligned}
    \right.
\end{equation}
\begin{prop}
\label{pr.flow}
Under the assumptions {\bf{(A1)}} and {\bf{(A2)}}, for all $z_{0} \in \mathcal Z$, the
previous Cauchy problem admits a unique mild solution $z_{t}$  in
$\mathcal{C}^{0}(\R,\mathcal Z)\cap{\mathcal{C}^{{1}}(\R,D(A)^{{'}})}$. \\ Furthermore the Cauchy problem:
\begin{equation}
\label{eq.hartree2}
\left\{
      \begin{aligned}
       \partial_{t}\tilde z_{t}&=v(t,\tilde z_{t})=-ie^{itA}[\partial_{\bar z}Q](e^{-itA}\tilde z_{t}),&\\
       \tilde z _{t=0}&=z_{0}.\\
         \end{aligned}
    \right.
\end{equation}
is equivalent with the initial problem and admits a unique solution $\tilde{z_{t}} \in  \mathcal{C}^{1}(\R,\mathcal Z)$. \\ Furthermore this equation implies:
\begin{itemize}
\item
$\forall t \in \R , \; |z_{t}|_{\mathcal Z}=|z_{0}|_{\mathcal
  Z}=|\tilde{z_{t}}|_{\mathcal Z}$,
 \item
 The velocity field $v(t,z)=-ie^{itA}[\partial_{\bar z}Q](e^{-itA}z)$, satisfies
 \begin{equation}
\label{eq.velocity}
 \forall t \in \R, \; |v(t,z)| \leq{r \; M (\sum_{j=2}^{r}|z|^{2j-1})         },
 \end{equation}
  with $M=\max_{j=2,...,r}\|\tilde{Q_{j}}\|$.
  \end{itemize}
\end{prop}

\begin{proof}
It is enough to consider only positive times $t>0$.
We will prove that $z \rightarrow v(t,z)=-ie^{itA}[\partial_{\bar
  z}Q](e^{-itA}z)$ is locally Lipschitz in $\mathcal Z$ which will give
the local existence and uniqueness on a time interval $[0,T^{*}[$ for
the equation \eqref{eq.hartree2}. Then we can recover solutions of the
original  equation \eqref{eq.hartree0} by setting $z_{t}=e^{-itA}\tilde z_{t}$.\\
Let $z,y$ be in $\mathcal Z$,
\begin{align*}
|v(t,z)-v(t,y)| & \leq{|[\partial_{\bar z}Q](e^{-itA}z)-[\partial_{\bar z}Q](e^{-itA}y)|} ,&\\
&\leq{\sum_{j=2}^{r}} j |\; (\langle z_{t}^{\otimes j-1 }| \vee \Id_{\mathcal Z})
\tilde{Q_{j}}(z_{t}^{\otimes j})-\langle y_{t}^{\otimes j-1 }| \vee
\Id_{\mathcal Z}) \tilde{Q_{j}}(y_{t}^{\otimes j}) \;|\,.
\end{align*}
Thus, by setting $M=\max_{j=2,...,r}\|\tilde{Q_{j}}\|$, for all
$z,y \in B(0,R)$,\; there exists a non negative constant $C_{R} > 0$
such that:
\begin{equation*}
|v(t,z)-v(t,y)|\leq{r \; C_{R} \; M |z-y|}.
\end{equation*}
 Thus the Cauchy-Lipschitz theorem gives a unique solution
 $\tilde{z_{t}}$ in  $\mathcal{C}^{1}([0,T^{*}[,\mathcal Z)$.
The previous calculus with $y=0$ gives, for $t \in [0,T^{*}[$, the estimate:
\begin{equation*}
|v(t,z)| \leq{M \; r \sum_{j=2}^{r}|z|^{2j-1}}.
\end{equation*}
It remains to prove $|z_{t}|=|z_{0}|$ for all $ t \in  [0,T^{*}[$
which ensures that $T^*=+\infty$. In fact
\begin{align*}
\partial_{t}|\tilde z_{t}|^{2}=2 \Real \langle \tilde
z_{t},\partial_{t}\tilde z_{t}\rangle
&= -2 \Real \langle \tilde z_{t}, ie^{itA}[\partial_{\bar z}Q](e^{-itA}\tilde z_{t})\rangle\\
&=-2 \Real \; i \langle e^{-itA}\tilde z_{t},[\partial_{\bar z}Q](e^{-itA}\tilde z_{t})\rangle\\
&=-2 \sum_{\ell=2}^{r}\Real \; i \langle e^{-itA}\tilde
z_{t},[\partial_{\bar z}Q_{\ell}](e^{-itA}\tilde z_{t})\rangle\\
&= -2 \sum_{\ell=2}^{r}\Real \;[i\ell
\underbrace{Q_{\ell}(e^{-itA}\tilde z_{t})}_{\in \R}]=0.
\end{align*}
So this shows that $|\tilde{z}_{t}|=|z_{t}|=|z_{0}|$ and the mass conservation is
proved. By setting $z_{t}=e^{-itA}\tilde{z}_{t}$ and using the fact that the solution
$\tilde{z}_{t}$ satisfies
\begin{equation*}
\tilde{z_{t}}=z_{0}-i\int_{0}^{t}e^{isA}[\partial_{\bar z}Q](e^{-isA}\tilde{z}_{s})ds,
\end{equation*}
we obtain
\begin{equation*}
z_{t}=e^{-itA} z_{0}-i\int_{0}^{t}e^{i(s-t)A}[\partial_{\bar z}Q](z_{s})ds.
\end{equation*}
Hence the function $t \mapsto z_{t}$ belongs to
$\mathcal{C}^0(\R,\mathcal Z)\cap{\mathcal{C}^{{1}}(\R,D(A)^{{'}})}$
and it is a mild solution of \eqref{eq.hartree0}.
\end{proof}

\section{Propagation of Wigner measures}

\subsection{The main convergence arguments}
The following proposition will be useful in the derivation of the transport equation. It is mainly due to the compactness of the $\tilde{Q_{j}}$'s.
\begin{prop}
\label{pr.compact}
Let $(\varrho_{\varepsilon})_{\varepsilon \in (0,\bar \varepsilon)}$  be
a family of normal states on $\Gamma_{s}(\mathcal Z)$. Assume the
Assumptions {\bf{(A1)}} and {\bf{(A2)}} are satisfied and
\begin{equation*}
\forall k \in \N, \; \exists C_{k}>0, \; \forall \varepsilon \in (0,\bar \varepsilon), \; \Tr[\varrho_{\varepsilon}\mathbf{N}^{k}] \leq C_{k}.
\end{equation*}
Assume furthermore that:
\begin{equation*}
\mathcal M (\varrho_{\varepsilon}, \varepsilon \in (0,\bar \varepsilon))=\{ \mu \}.
\end{equation*}
Then for all $\xi \in \mathcal Z$ and all $t \in \R$:
\begin{equation}
\label{eq.toto1}
\lim_{\varepsilon \longrightarrow
  0}\Tr[\varrho_{\varepsilon}W(\sqrt{2}\pi\xi) \{ \mathrm{D}[Q(e^{-itA}z)][\xi]
\}^{Wick}]=\int_{\mathcal Z}e^{2i \pi \Real \langle\xi,z\rangle}\mathrm{D}[Q(e^{-itA}z)][\xi] d\mu(z),
\end{equation}
where $\mathrm{D}[Q(e^{-itA}z)][\xi]=\langle[\partial_{\bar
  z}Q](e^{-itA}z),e^{-itA}\xi\rangle+\langle e^{-itA}\xi,[\partial_{\bar z}Q](e^{-itA}z)\rangle.$
\end{prop}

\begin{proof}
For $j \in \{2,\ldots r\}$ and $\xi \in \mathcal Z$, let $B_{j}(\xi)$ denote the operator
\begin{equation}
B_{j}(\xi)=\tilde{Q_{j}}(\Id_{\bigvee^{j-1}\mathcal Z}\otimes |\xi \rangle) ,
\end{equation} and
\begin{equation*}
B_{j}^{*}(\xi)=(\Id_{\bigvee^{j-1}\mathcal Z} \otimes\langle \xi| )\tilde{Q_{j}}.
\end{equation*}
Both operators are compact respectively from $\bigvee^{j-1} \mathcal
Z$ to $\bigvee^{j} \mathcal Z$ and from $\bigvee^{j} \mathcal Z$ to
$\bigvee^{j-1} \mathcal Z$ owing to the assumption {\bf{(A2)}}.
Now, let us check that $\mathrm{D}[Q(e^{-itA}z)][\xi]$ is the sum of symbols with
compact kernels. Actually, $Q(e^{-itA}z)=\sum_{j=2}^{r}\langle
z^{\otimes j}, (e^{itA})^{\otimes j}\tilde{Q}_{j}(e^{-itA})^{\otimes
  j}z^{\otimes j}\rangle$ with $\tilde{Q}_{j}^{*}=\tilde{Q}_{j}.$ In
particular with $\overline{Q(z)}=Q(z),$ we obtain
\begin{align*}
\mathrm{D}[Q(e^{-itA}z)][\xi]&=\langle [\partial_{\bar z}Q](e^{-itA}z),e^{-itA}\xi \rangle +\langle e^{-itA}\xi,[\partial_{\bar z}Q](e^{-itA}z)\rangle,&\\
&=\sum_{j=2}^{r}j[\langle z^{\otimes j},\tilde{Q_{j}}(e^{-itA}\xi \vee z^{\otimes j-1})\rangle +\langle\tilde{Q_{j}}(e^{-itA}\xi \vee z^{\otimes j-1}),z^{\otimes j}\rangle],&\\
&=\sum_{j=2}^{r}j[\langle z^{\otimes j},B_{j}(e^{-itA}\xi)z^{\otimes j-1}\rangle+\langle z^{\otimes j-1},B_{j}^{*}(e^{-itA}\xi)z^{\otimes j}\rangle],&
\end{align*}
and all the terms involve compact operators. We refer to Lemma
\ref{eq.wigcompact} in order to compute the limit of $\Tr[\varrho_{\varepsilon}W(\sqrt{2}\pi\xi) \{ \mathrm{D}[Q(e^{-itA}z)][\xi]
\}^{Wick}]$ and obtain \eqref{eq.toto1}.

\end{proof}

In order to understand the asymptotic behaviour of
$\Tr[\rho_{\varepsilon}(t)W(\sqrt{2}\pi\xi)]$, when $\varepsilon$ goes to
$0$, we shall prove the following proposition.
\begin{prop}
\label{pr.compact2}
Let $(\varrho_{\varepsilon})_{\varepsilon \in (0,\bar \varepsilon)}$  be
a family of normal states on $\Gamma_{s}(\mathcal Z)$. Assume the
assumptions {\bf{(A1)}} and {\bf{(A2)}} are satisfied and
\begin{equation*}
\forall k \in \N, \; \exists C_{k}>0, \; \forall \varepsilon \in (0,\bar \varepsilon), \; \Tr[\varrho_{\varepsilon}\mathbf{N}^{k}] \leq C_{k}.
\end{equation*}
Assume furthermore that
\begin{equation*}
\mathcal M (\varrho_{\varepsilon}, \varepsilon \in (0,\bar \varepsilon))=\{ \mu \}.
\end{equation*}
Then for all $\xi \in \mathcal Z$ and all $t \in \R$,
\begin{equation*}
\lim_{\varepsilon \to 0} \int_{0}^{t}{\Tr[\varrho_{\varepsilon}W(\sqrt{2}\pi\xi) \{ \sum_{j=2}^{r}\varepsilon^{j-1}\frac{(i\pi)^{j}}{j!}\mathrm{D}^{j}[Q(e^{-isA}z)][\xi] \}^{Wick}] ds}\,=0.
\end{equation*}
\end{prop}
\begin{proof}
By using Proposition \eqref{number}, a simple estimate of  the integrand
yields for all $s \in [0,t]$
\begin{align*}
&|\Tr[\varrho_{\varepsilon}W(\sqrt{2}\pi\xi) \{
\sum_{j=2}^{r}\varepsilon^{j-1}\frac{(i\pi)^{j}}{j!}\mathrm{D}^{j}[Q(e^{-isA}z)][\xi]
\}^{Wick}]|\\
 &\leq C_{r}  \sum_{j=2}^{r}
\varepsilon^{j-1}\frac{\pi^{j}}{j!}\|\langle\mathbf N \rangle^{-r}\{\mathrm{D}^{j}[Q(e^{-isA}z)][\xi] \}^{Wick}\|\\
&\leq \sum_{j=2}^{r}
\varepsilon^{j-1}\frac{\pi^{j}}{j!}\tilde{C_{r}}\langle \xi\rangle^{j},
\end{align*}
with $\langle u\rangle=(1+|u|^{2})^{\frac{1}{2}}$. We conclude
therefore by the dominated convergence theorem.
\end{proof}

\subsection{Existence of Wigner measures for all times}
Remember the definition of
\begin{equation*}
\varrho_{\varepsilon}(t)=e^{-i\frac{t}{\varepsilon}H_{\varepsilon}}\varrho_{\varepsilon}e^{i\frac{t}{\varepsilon}H_{\varepsilon}},
\end{equation*}
and
\begin{equation*}
\tilde{\varrho_{\varepsilon}}(t)=e^{i\frac{t}{\varepsilon}H_{\varepsilon}^{0}}\varrho_{\varepsilon}(t)e^{-i\frac{t}{\varepsilon}H_{\varepsilon}^{0}}.
\end{equation*}
\begin{prop}
\label{prop.existence}
Let $(\varrho_{\varepsilon})_{\varepsilon \in (0,\bar \varepsilon)}$
be a family of normal states on $\Gamma_s(\mathcal{Z})$. Assume the
assumptions {\bf{(A1)}}-{\bf{(A2)}} are satisfied and
\begin{equation}
\forall k \in \N, \exists C_{k} > 0, \; \forall \varepsilon \in (0,\bar \varepsilon), \; \Tr[\varrho_{\varepsilon}\mathbf{N}^{k}] \leq C_{k}\,.
\end{equation}
For all sequence $(\varepsilon_{n})_{n \in \N}$ in $(0,\bar
\varepsilon)$ such that $\lim_{n \to
  +\infty}\varepsilon_{n}=0$, there exist a subsequence
$(\varepsilon_{n_{k}})_{k \in \N}$ with $\lim_{k \to
  +\infty}\varepsilon_{n_{k}}=0$ and a family of Borel probability
measures $\{\tilde{\mu_{t}}, t \in \R\}$ such that:
\begin{equation}
\forall t \in \R, \; \mathcal M ( \tilde{\varrho}_{\varepsilon_{k}}(t), k \in \N) = \{  \tilde{\mu_{t}}  \}.
\end{equation}
Furthermore,
\begin{equation}
\label{eq.kmasse}
\int_{\mathcal Z}|z|^{2k}\; d\tilde{\mu_{t}}(z) \leq C_{k}, \; \forall k \in \N,
\end{equation}
and $\tilde{\mu}_{t}$ solves the integral equation, for all $\xi \in
\mathcal Z$,
\begin{equation*}
\tilde{\mu_{t}}(e^{2i\pi \Real\langle\xi,.\rangle})  = \tilde{\mu_{0}}(e^{2i \pi
  \Real\langle\xi,.\rangle} )-\pi  \int_0^t \tilde{\mu_{s}} (e^{2i \pi
  \Real\langle\xi,.\rangle}\mathrm{D}[Q(e^{-isA}.)][\xi])ds,
\end{equation*}
\begin{equation}
\label{eq.mesureliouville}
= \tilde{\mu_{0}} (e^{2i \pi \Real\langle\xi,.\rangle})+i \int_0^t \tilde{\mu_{s}}(
\{ \tilde{Q_{s}}\; , \;e^{2i \pi \Real\langle\xi,.\rangle}(z) \}) ds,
\end{equation}
by setting for $b_{1}$,$b_{2} \in \mathcal{P}_{p,q}(\mathcal Z)$
\begin{equation*}
\{ b_{1},b_{2} \} (z) = \partial_{z}b_{1}(z).
\partial_{\bar z} b_{2}(z)-\partial_{z}b_{2}(z).\partial_{\bar z}b_{1}(z).
\end{equation*}
\end{prop}
\begin{proof}
The extraction of such subsequence $(\varepsilon_{n_{k}})_{k \in \N}$
and  the existence of  a family of Borel probability measures $\tilde{\mu_{t}}$ have been proved
 in \cite{AmNi4} by a diagonal extraction  process relying on some Ascoli type argument. We skip the proof
of this step since the result in \cite{AmNi4}  applies to our case without modification. \\
Let $p_{n}$ be the projection on $\C e_{1}\oplus \ldots \oplus
\C e_{n}$ with $(e_{i})_{i\in\mathbb{N}}$ an ONB of $\mathcal{Z}$. Since
$\mathrm{d}\Gamma(p_{n}) \leq \mathbf{N}\,$, it follows that
\begin{equation*}
\int_{\mathcal Z}|z|^{2k} d\tilde\mu_{t}(z)=\sup_{n \in \N}\int_{\mathcal
  Z}|p_{n}z|^{2k} d\tilde\mu_{t}(z)=\sup_{n \in \N}\{\liminf_{\varepsilon
  \to 0} \Tr [\tilde\varrho_{\varepsilon}(t){(\mathrm{d}\Gamma(p_{n}))}^{k}]\}
\end{equation*}
\begin{equation*}
 \leq \liminf_{\varepsilon \to 0} \Tr [\varrho_{\varepsilon}\mathbf{N}
^{k}] \leq C_{k}\;.
\end{equation*}
This proves \eqref{eq.kmasse}. For the derivation of the integral
equation \eqref{eq.mesureliouville}, we have  according to \eqref{eq.ipp},
\begin{equation*}
\Tr[\tilde{\varrho_{\varepsilon}}(t)W(\sqrt{2}\pi\xi)] = \Tr[\varrho_{\varepsilon}W(\sqrt{2}\pi\xi)] +
\end{equation*}
\begin{equation*}
  i\int_{0}^{t}{\Tr[\tilde{\varrho_{\varepsilon}}(s)W(\sqrt{2}\pi\xi) \{ \sum_{j=1}^{r}\varepsilon^{j-1}\frac{(i\pi)^{j}}{j!}\mathrm{D}^{j}[Q(e^{-isA}z)][\xi] \}^{Wick}] ds}.
\end{equation*}
The estimate of Proposition \ref{pr.compact2} implies that all the terms $j=2,...,r$  of the sum in the right side go to $0$ as $\varepsilon \rightarrow 0$.
For the last term, we use Proposition \ref{pr.compact} for \
$\tilde{\varrho_{\varepsilon}}(s)$ thanks to the fact that
\begin{equation*}
\Tr[\tilde{\varrho}_{\varepsilon}(t) \mathbf{N}^{k}]=\Tr[\varrho_{\varepsilon}\mathbf{N}^{k}] \leq{C_{k}}\,.
\end{equation*}
Thus taking the limit as $\varepsilon \rightarrow 0$ yields
\begin{equation*}
\tilde{\mu_{t}}(e^{2i \pi \Real\langle\xi,.\rangle})  =
\tilde{\mu_{0}}(e^{2i \pi \Real\langle\xi,.\rangle} )-\pi  \int_0^t \tilde{\mu_{s}}
(e^{2i \pi \Re\langle\xi,.\rangle}\mathrm{D}[Q(e^{-isA}.)][\xi])ds.
\end{equation*}
We conclude with
\begin{align*}
i\{ \tilde{Q_{s}}\,,\,e^{2i\pi\Real\langle\xi,.\rangle}
\}(z)&=i(\langle [\partial_{\bar
  z}Q](e^{-isA}z)\,,\,\partial_{\bar z}e^{2i \pi\Real\langle\xi,z\rangle}\rangle-\langle\partial_{\bar
  z}e^{2i \pi\Real\langle\xi,z\rangle}\,,\,[\partial_{\bar z}Q](e^{-isA}z)\rangle),&\\
&=i\pi(\langle [\partial_{\bar z}Q](e^{-isA}z)\,,\,i\xi\rangle-\langle i
\xi,[\partial_{\bar z}Q](e^{-isA}z)\rangle) \,e^{2i \pi \Real\langle\xi,z\rangle}\,,&\\
&=-\pi \mathrm{D}[Q(e^{-isA}z)][\xi] \,e^{2i \pi \Real\langle\xi,z\rangle}\,.
\end{align*}
\end{proof}

\subsection{The Liouville equation fulfilled by the Wigner measures.}
The previous integral equation \eqref{eq.mesureliouville} can be
interpreted as a continuity equation,  in the infinite dimensional Hilbert
space  $\mathcal Z$, fulfilled by the Wigner measures $(\tilde\mu_{t})$.  \\
We introduce some classes of cylindrical functions on $\mathcal Z$. Denote $\P$ the space of the
finite rank orthogonal projections on $\mathcal Z$. A function is
in the cylindrical Schwarz space $\mathcal S_{cyl}(\mathcal Z)$ (resp.
$C_{0,cyl}^{\infty}(\mathcal Z)$) if
 \begin{equation*}
\exists p \in \P, \; \exists g \in  \mathcal S(p \mathcal
Z)(\text{resp. $C_{0,cyl}^{\infty}(p\mathcal Z)$}),\; \forall z \in \mathcal Z, f(z)=g(pz).
\end{equation*}
The space $C_{0,cyl}^{\infty}(\R \times \mathcal Z)$ which enforces
the compact support in the first variable, will be useful too.
Denote $L_{p}(dz)$ the Lebesgue measure associated with the finite
dimensional subspace $p\mathcal Z$. The Fourier transform is given on $\mathcal S_{cyl}(\mathcal Z)$ by :
\begin{equation*}
\mathcal F[f](\xi)=\int_{p\mathcal Z} f(z) e^{-2i\pi \Real
  \langle z,\xi \rangle_{\mathcal Z}}L_{p}(dz),
\end{equation*}
\begin{equation*}
f(z)=\int_{p \mathcal Z} \mathcal F [f](\xi) e^{2i\pi \Real
  \langle z,\xi \rangle_{\mathcal Z }}L_{p}(d\xi).
\end{equation*}
Then call $Prob_{2}(\mathcal Z)$ the set of Borel probability
measures $\mu$ finite second moment, i.e. $\int_{\mathcal Z} |z|_{\mathcal Z}^{2} \; d\mu(z) <
\infty$. On this space the Wasserstein distance is given by the formula:
\begin{equation}
W_{2}(\mu_{1},\mu_{2})=\sqrt{\inf_{\mu \in \Gamma(\mu_{1},\mu_{2})}
\int_{\mathcal Z} | z_{1} - z_{2} |_{\mathcal Z}^{2} d\mu(z_{1},z_{2})},
\end{equation}
with $\Gamma(\mu_{1},\mu_{2})$ the set of probability measures $\mu$ on
$\mathcal Z\times \mathcal Z$ such that the marginals
$(\Pi_{1})_{*}\mu=\mu_{1}$ and $(\Pi_{2})_{*}\mu=\mu_{2}$.
Let $\mathcal{P}(\mathcal Z)$ be the family of all Borel probability
measures on a Hilbert space $\mathcal Z$. Here $\Pi_{j}, j=1,2$, are
the canonical projections on the first and the second component
respectively.\\

From now, after
introducing a Hilbert basis $(e_{n})_{n \in \N^{*}}$, the space
$\mathcal{Z}$ can be equipped with the
distance $$d_{w}(x_{1}-x_{2})=\sqrt{\sum_{n \in \N^{*}}\frac{|\langle
      x_{1}-x_{2},e_{n}\rangle|^{2}}{n^{2}}}.$$ It induces a topology
  globally weaker than the weak topology. However
  these topology coincide on bounded sets of $\mathcal{Z}.$\\
The norm and $d_{w}$ topology give rise two distinct notions
of narrow convergence  of probability measures. On the one hand, a sequence $(\mu_{n})_{n \in \N}$ is
narrowly convergent to $\mu \in \mathcal P(\mathcal Z)$ if 
\begin{equation}
\label{narrow1}
\lim_{n \to
  +\infty}\int_{\mathcal Z}f(z)d\mu_{n}(z)=\int_{\mathcal
  Z}f(z)d\mu(z),
\end{equation}
for every function $f \in
\mathcal{C}^{0}_{b}(\mathcal Z,|.|)$, the space of continuous and bounded
real functions defined on $\mathcal Z$ with the norm
topology. On the other hand, a sequence $(\mu_{n})_{n \in \N}$ is
weakly narrowly convergent if  the limit \eqref{narrow1} holds for
all $f \in \mathcal{C}^{0}_{b}(\mathcal Z,d_{w})$. Our sequences of
probability measures are assumed to have a uniformly bounded moment $\int_{\mathcal
  Z}|z|^{2k}d\mu_{n}(z)\leq C_{k}$ for some $k\geq 1$. Within this framework, the narrow convergence is equivalent to the convergence with respect
to the Wasserstein distance $W_{2}$ in $Prob_{2}(\mathcal Z)$
according to Proposition 7.1.5 in \cite{AGS}. With the same moment
condition, the weak narrow convergence is equivalent to the
convergence \eqref{narrow1} for all $f \in \mathcal{S}_{cyl}(\mathcal
Z)$ or for all $f \in \mathcal{C}^{\infty}_{0,cyl}(\mathcal Z),$
acccording to Lemma 5.1.12 f) in \cite{AGS}.
\begin{prop}
\label{prop.liouville}
Assume that the family $(\varrho_{\varepsilon})_{\varepsilon \in (0,\bar{\varepsilon})}$ satisfies:
\begin{equation}
\forall \alpha \in \N, \, \exists \;C_{\alpha}>0 \,, \forall \varepsilon \in (0,\bar{\varepsilon}), \, \Tr[\varrho_{\varepsilon} \mathbf{N}^{\alpha}] \leq{C_{\alpha}}.
\end{equation}
Consider a subsequence $(\varepsilon_{k})_{k \in \N}$ such that
$\lim_{k \to +\infty}\varepsilon_{k}=0$,
\begin{equation}
\mathcal{M}(\tilde{\varrho}_{\varepsilon_{k}}(t), k \in \N) = \{ \tilde{\mu_{t}} \}.
\end{equation}
Then the probability measure $\tilde{\mu_{t}}$ defined on $\mathcal Z$ satisfies
\begin{enumerate}
\item
When $(e_{n})_{n \in \N^{*}}$ a Hilbert basis of $\mathcal Z$
and $\mathcal Z$ is endowed with the distance
$d_{w}(z_{1},z_{2})=\sqrt{\sum_{n \in
    \N^{*}}\frac{|<z_{1}-z_{2},e_{n}>| }{n^{2} } } $, the measure
$\tilde{\mu_{t}}$ is weakly narrowly continuous with respect to t.
\item
This is a weak solution to the (continuity) Liouville equation
\begin{equation}
\partial_{t} \tilde {\mu}_{t} + i \{ Q_{t},\tilde{\mu}_{t} \} = 0,
\end{equation}
in the sense that for all $f \in C_{0,cyl}^{\infty}(\R\times\mathcal Z)$
\begin{equation}
\int_{\R}\int_{\mathcal Z} (\partial_{t} f + i \{ Q_{t},f \} ) d\tilde{\mu}_{t}(z) dt =0,
\end{equation}
with $Q_{t}(z)=Q(e^{-itA}z).$
\end{enumerate}
\end{prop}
\begin{proof}
\noindent\textbf{a)}
The characteristic function $G$ of the measure $\tilde\mu_{t}$ is given by
$$G(\eta,t)=\tilde{\mu}_{t}(e^{-2i\pi \Re\langle\eta,z\rangle})\,.$$
 The following inequality holds:
\begin{equation}
\label{eq.GG}
|G(\eta,t)-G(\eta',t)|\leq{    2\pi   |\eta-\eta'| \int_{\mathcal Z} |z|    d\tilde{\mu_{t}}(z) .     }
\end{equation}
Since the uniform estimate $\int_{\mathcal Z}1+|z|^{2}_{\mathcal Z}
d\tilde{\mu}_{t}(z) \leq{C_{2}}$  is true for all times, we get for all $\eta,\eta'$ in $\mathcal Z$ and for $t \in \R$,
\begin{equation}
|G(\eta,t)-G(\eta',t)| \leq \pi |\eta-\eta'|C_{2}.
\end{equation}
\noindent\textbf{b)}
According to Proposition \ref{prop.existence} and \eqref{eq.mesureliouville},
\begin{equation*}
\tilde{\mu}_{t'}(e^{2i\pi \Re\langle\xi,.\rangle})-\tilde{\mu}_{t}(e^{2i \pi
  \Re\langle\xi,.\rangle})=-\pi \int_{t}^{t'} \tilde{\mu_{s}} (e^{2i
  \pi\Re \langle\xi,.\rangle}\mathrm{D}[Q(e^{-isA}.)][\xi])ds.
\end{equation*}
We use the estimate \eqref{eq.velocity} and get
\begin{equation*}
|\mathrm{D}[Q(e^{-isA}z)][\xi]| \leq{2 |e^{-isA}\xi| \; |[\partial_{\bar{z}}Q](e^{-isA}z)|}
\leq{2|\xi| M r \sum_{j=2}^{r}|z|^{2j-1}}.
\end{equation*}
Thus for $\xi \in \mathcal Z$
\begin{align*}
|G(\xi,t')-G(\xi,t)|& \leq \Big| \pi \int_{t}^{t'} G(\xi,s)
\mathrm{D}[Q(e^{-isA}z)][\xi] ds \Big| ,&\\
&\leq { 2 \pi  |t-t'| |\xi| M r \sup_{s \in
    [t,t']}\sum_{j=2}^{r}\int_{\mathcal Z}|z|^{2j-1}
  d\tilde{\mu}_{s}(z)},&\\
&\leq C|t-t'||\xi|,&
\end{align*}
since $|z|^{2j-1} \leq \frac{1}{2}(1+|z|^{2(2j-1)}) \in L^{1}(\mathcal
Z,\tilde{\mu}_{t})$ and $$\sum_{j=2}^{r}\int_{\mathcal Z}|z|^{2j-1}
d\tilde{\mu}_{s}(z)\leq C_{r},$$ with a time independent constant  $C_{r}$.
Hence for all $\xi$ in $\mathcal Z$ and for all $t,t'$ in $\R$:
\begin{equation}
\label{lipschitzG}
|G(\xi,t')-G(\xi,t)| \leq 2\pi C_{r} M r |t-t'||\xi|.
\end{equation}
\begin{enumerate}
\item
Take now $g \in \mathcal S_{cyl}(\mathcal Z)$ based on $p \mathcal Z$
and the equality holds:
\begin{equation}
I_{g}(t)=\int_{\mathcal Z} g(z) d\tilde{\mu}_{t}(z)=\int_{p \mathcal Z}\mathcal F[g](\eta)G(\eta,t)dL_{p}(\eta).
\end{equation}
We shall establish the continuity of $I_{g}$ on $\R$. Indeed
\begin{itemize}
\item
$t \longrightarrow \mathcal F[g](\eta)G(\eta,t)$ is continue
owing to \eqref{lipschitzG}
\item
$\eta \longrightarrow \mathcal F[g](\eta)G(\eta,t)$ is bounded
by a $L_{p}(d\eta)$-integrable function thanks to
\eqref{eq.GG} and $\mathcal F[g] \in \mathcal S(p\mathcal Z)$.\\
\end{itemize}
\bigskip
Thus we have the continuity of $I_{g}$ for all $g \in \mathcal
S_{cyl}(\mathcal Z)$. Furthermore the uniform estimate
condition $$\forall \alpha \in \N, \; \int_{\mathcal Z}|z|^{2\alpha}
d\tilde{\mu_{t}}(z) \leq C_{\alpha},$$ with $C_{\alpha}$ time independent allow us to apply lemma 5.1.12-f) in
\cite{AGS} and to assert that the map $t \rightarrow \tilde{\mu}_{t}$
is weakly narrowly continuous.

\item
We integrate the expression \eqref{eq.mesureliouville} with respect to $\mathcal
F[g](\eta)L_{p}(dz)$:
\begin{equation*}
\forall t \in \R ,\, \forall g \in \mathcal S_{cyl}(\mathcal Z), \int_{\mathcal Z}g(z) d\tilde{\mu}_{t}(z)=\int_{\mathcal
  Z}g(z)d\tilde{\mu}_{0}(z)+i\int_{0}^{t}\int_{\mathcal Z}\{ Q_{s},g\} d\tilde{\mu}_{s}(z) ds.
\end{equation*}
Hence $I_{g}$ belongs to $C^{1}(\R)$ and satisfies:
\begin{equation*}
\partial_{t}I_{g}(t)=i\int_{\mathcal Z}\{Q_{t},g\}(z) d\tilde{\mu}_{t}(z).
\end{equation*}
Multiplying this expression by a function $\phi \in
C^{\infty}_{0}(\R)$ and integrating by parts lead to
$$\int_{\R}\partial_{t}I_{g}(t)\phi(t)dt=i\int_{\R \times \mathcal
  Z}\{Q_{t},g\}(z) d\tilde{\mu}_{t}(z)\phi(t)dt .$$
Integrating by parts gives
\begin{align*}
\int_{\R \times \mathcal Z}g(z)
d\tilde{\mu}_{t}(z)\phi'(t)dt+i\int_{\R \times \mathcal
  Z}\{Q_{t},g\}\phi(t) d\tilde{\mu}_{t}(z)dt=0,\\
\mbox{or} \int_{\R \times \mathcal Z}(\partial_{t}f(t,z)+i \{ Q_{t},f\}) d\tilde{\mu}_{t}(z)dt=0,
\end{align*}
with $f(t,z)=g(z)\phi(t)$.
\bigskip
\\
We conclude by using the density of $ C^{\infty}_{0}(\R)
\otimes^{alg} C^{\infty}_{0,cyl}(\mathcal Z)
$ in $C^{\infty}_{0,cyl}(\R \times \mathcal Z).$

\end{enumerate}
\end{proof}

\subsection{Convergence toward the mean field dynamics}

\begin{prop}
\label{pr.uniqueness}
Assume that the family of normal states
$(\varrho_{\varepsilon})_{\varepsilon \in (0,\bar{\varepsilon})}$ on
$\Gamma_{s}(\mathcal Z)$
fulfills the assumptions {\bf{(A1)}}-{\bf{(A2)}}, with the uniform control
\begin{equation*}
\forall \alpha \in \N, \, \exists C_{\alpha}>0 \,, \forall \varepsilon \in (0,\bar{\varepsilon}), \; \Tr[\varrho_{\varepsilon}\mathbf{N}^{\alpha}] \leq{C_{\alpha}},
\end{equation*}
and
\begin{equation*}
\mathcal{M}(\varrho_{\varepsilon}, \varepsilon \in (0,\bar{\varepsilon})) = \{ \mu_{0} \}.
\end{equation*}
Then for any time $t \in \R$ the family
$(\varrho_{\varepsilon}(t)=e^{-i\frac{t}{\varepsilon}H_{\varepsilon}}\varrho_{\varepsilon}e^{i\frac{t}{\varepsilon}H_{\varepsilon}})_{\varepsilon
  \in (0,\bar{\varepsilon})}$ admits a unique Wigner measure $\mu_{t}$
equal to $\Phi(t,0)_{*}\mu_{0}$, where $\Phi$ is the flow associated with the well defined Hartree
equation owing to Proposition \ref{pr.flow}. Moreover, the map
$t\mapsto \mu_{t}\in Prob_{2}(\mathcal Z)$ is  continuous with respect
to  the Wasserstein distance $W_{2}$.
\end{prop}

\begin{proof}
Take $\tilde{\varrho_{\varepsilon}}(t)$ as in  \eqref{eq.deftilderho}
and consider  the Hartree equation \eqref{eq.hartree0} with the flow $\Phi(t,s)$ corresponding to:
\begin{equation}
i\partial_{t} z_{t} = \partial_{\bar z}h(z,\bar z)=Az_{t}+\partial_{\bar z}Q({z_{t}}),
\end{equation}
on $\mathcal Z$ and the flow $\tilde{\Phi}(t,s)$ associated with
\begin{equation*}
 \partial_{t}\tilde{z}_{t}=v(t,\tilde{z}_{t}) \quad \mbox{with} \quad v(t,z)=-ie^{itA}[\partial_{\bar z}Q](e^{-itA}z)\,.
\end{equation*}
Proposition \ref{pr.flow} provides the following estimate
\begin{equation*}
|v(t,z)|_{\mathcal Z} \leq{ M r \sum_{j=2}^{r}|z|_{\mathcal Z}^{2j-1}}\,,
\end{equation*}
with $M=\max_{j=2,...,r}\|\tilde{Q}_{j}\|.$
Recall that $(\tilde{\mu}_{t})$ are the Wigner measures defined for
all times and associated with a subsequence $(\tilde\rho_{\varepsilon_{n_{k}}})_{k
  \in \N}$, hence we obtain
\begin{align*}
|v(t,z)|_{L^{2}(\mathcal Z,\tilde{\mu}_{t})}  &= \sqrt{\int_{\mathcal Z}|v(t,z)|^{2}d\tilde{\mu}_{t}(z)} &\\
&\leq M r\sqrt{\sum_{j=2}^{r}\int_{\mathcal Z}|z|^{2(2j-1)}d\tilde{\mu}_{t}}(z) \in L^{1}([-T,T]).
\end{align*}
This holds since $$\forall j \in \N, \; \int_{\mathcal Z}|z|^{2j}d\tilde{\mu}_{t}(z)
\leq{C_{j}}\,,$$ with $C_{j}$ time independent. \\
Now, using Proposition \ref{prop.liouville} the measure
$\tilde{\mu}_{t}$ satisfies
\begin{equation*}
\partial_{t} \tilde {\mu}_{t} + i \{ Q_{t},\tilde{\mu}_{t} \} = \partial_{t}\tilde{\mu}_{t} + \nabla^{T}(v(t,z)\tilde{\mu}_{t})=0,
\end{equation*}
in the weak sense and the map $t \mapsto \tilde{\mu}_{t}\in Prob_{2}(\mathcal Z)$ is
weakly narrowly continuous. Moreover, the velocity field $v(t,.)$
satisfies the condition  $|v(t,z)|_{L^{2}(\mathcal  Z,\mu_{t})}$ belongs to $L^{1}([-T,T])$.
\\
Thus, $\tilde{\mu}_{t}$ verifies the conditions of Proposition
\ref{prop.C1} (see \cite{AmNi4} for more details) with $I=[-T,T]$ and
then $\tilde{\mu}_{t}$ is continuous with respect to the Wasserstein
distance $W_{2}$. So now the measures $\tilde{\mu}_{t}$ satisfy  all the hypotheses
of  Proposition \ref{prop.C4}, i.e.:
\begin{itemize}
\item
 $t \mapsto\tilde{\mu}_{t} \in Prob_{2}(\mathcal Z)$ is $W_{2}$-continuous.
\item
For all $T>0$, $|v(t,z)|_{L^{2}(\mathcal
  Z,\mu_{t})}$ belongs to $L^{1}([-T,T])$.
\item
$\tilde{\mu_{t}}$ is the weak solution to:
\begin{equation*}
\partial_{t}\tilde{\mu}_{t}+\nabla^{T}(v(t,z)\tilde{\mu}_{t})=0,
\end{equation*}
\end{itemize}
subsequently $\tilde{\mu}_{t}=\tilde{\phi}(t,0)_{*}\mu_{0}$ and
\begin{equation*}
\mathcal M(\tilde{\varrho}_{\varepsilon}(t), \varepsilon \in (0,\bar \varepsilon)) = \{ \tilde{\mu}_{t}\},
\end{equation*}
for any time $t \in \R.$
By noticing that
$$\varrho_{\varepsilon}(t)=e^{-it\frac{t}{\varepsilon}H^{0}_{\varepsilon}}\varrho_{\varepsilon}e^{it\frac{t}{\varepsilon}H^{0}_{\varepsilon}},$$
we get
\begin{equation*}
\mathcal M(\varrho_{\varepsilon}(t), \varepsilon \in (0,\bar \varepsilon)) = \{ \mu_{t}\},
\end{equation*}
which finish the proof of Theorem \ref{thm.main} for regular data.

\end{proof}
\subsection{Evolution of the Wigner measure for general data}
In this subsection, we shall prove  Theorem \ref{thm.main} for
general data. We used the same truncation scheme used in
\cite{AmNi4}. Hence consider a family
$(\varrho_{\varepsilon})_{\varepsilon \in (0,\bar{\varepsilon})}$
satisfying the assumption of  Theorem \ref{thm.main}, i.e.:
$$
\exists \delta>0\,, \exists C_{\delta}>0\,,  \forall \varepsilon\in
\mathcal{E}\,,\quad
\Tr[\varrho_{\varepsilon}\mathbf{N}^{\delta}] \leq C_{\delta}<\infty\,.
$$
There exists another family $(\varrho_{\varepsilon}^{(m)})_{m \in \N}$
such that $\Tr[\varrho_{\varepsilon}^{(m)}]=1$,
\begin{equation*}
\forall k \in \N\,, \exists C_{k}>0\,,  \forall \varepsilon\in
(0,\overline{\varepsilon})\,,\quad
\Tr[\varrho_{\varepsilon}^{(m)}\mathbf{N}^{k}] \leq C_{k}<\infty
\end{equation*}
and   
\begin{equation}
\label{eq.troncature}
\lim_{m \to +\infty}\sup_{\varepsilon \in (0,\bar{\varepsilon})}\|\varrho_{\varepsilon}-\varrho_{\varepsilon}^{(m)}\|_{\mathcal{L}^{1}}=0.
\end{equation}
Indeed by setting 
\begin{equation*}
\rho_{\varepsilon}^{(m)}=\frac{1}{\Tr[\chi_{m}(\mathbf{N})\rho_{\varepsilon}\chi_{m}(\mathbf{N})]}\chi_{m}(\mathbf{N})\rho_{\varepsilon}\chi_{m}(\mathbf{N}),
\end{equation*} 
with $\chi_{m}(n)=\chi(\frac{n}{m})$ and $0\leq \chi
\leq 1, \chi \in C_{0}^{\infty}(\R)$ and $\chi \equiv 1$ in a
neighborhood of 0, the result \eqref{eq.troncature} holds.

The family $(\varrho_{\varepsilon}^{(m)})_{m \in \N}$ is satisfying the
assumption of  Proposition \ref{pr.uniqueness} and then by extracting a
subsequence $(\varepsilon_{n_{k}})_{k \in \N}$ such that for all $t
\in \R$
$$\mathcal{M}(\varrho_{\varepsilon_{n_{k}}}^{(m)}(t),k \in \N)=\{
\phi(t,0)_{*}\mu_{0}^{(m)}\},$$
with $\phi(t,0)$ the flow of the Hartree equation
\eqref{eq.hartree0}, and $\mu_{0}^{(m)}$ the Wigner measure associated with $\varrho_{\varepsilon}^{(m)}$.
Hence by setting $\mu_{t} \in
\mathcal{M}(\varrho_{\varepsilon}(t),\varepsilon \in
(0,\bar{\varepsilon}))$, there exists a subsequence
$(\varepsilon_{l})_{l \in \N}$ such that
$$\mathcal{M}(\varrho_{\varepsilon_{l}}^{(m)}(t),l \in \N)=\{
\mu_{t}\}.$$
Then a computation of the total variation by the triangle inequality
gives
\begin{align*}
\int_{\mathcal{Z}}|\mu_{t}-\phi(t,0)_{*}\mu_{0}| &\leq
\int_{\mathcal{Z}}|\mu_{t}-\phi(t,0)_{*}\mu_{0}^{(m)}|+\int_{\mathcal{Z}}|\phi(t,0)_{*}\mu_{0}^{(m)}-\phi(t,0)_{*}\mu_{0}|\\
&\leq \int_{\mathcal{Z}}|\mu_{t}-\phi(t,0)_{*}\mu_{0}^{(m)}|+\int_{\mathcal{Z}}|\mu_{0}^{(m)}-\mu_{0}|.
\end{align*}
By taking the limit when $m \to +\infty$ we get
$\int_{\mathcal{Z}}|\mu_{t}-\phi(t,0)_{*}\mu_{0}|=0$, hence
\begin{equation*}
\mathcal{M}({\varrho_{\varepsilon}}(t), \varepsilon \in (0,\bar{\varepsilon}))=\{\phi(t,0)_{*}\mu_{0}\}.
\end{equation*}
\section{Examples}
In this section we will give some examples of states which do not satisfy
the (PI) condition introduced in the paper \cite{AmNi3}. So that the
results in \cite{AmNi2, AmNi3}  can not be applied on these examples
however they satisfy the assumptions of our Theorem \ref{thm.main}.  We recall
the (PI) assertion below.
\begin{dfn}
Let $(\varrho_{\varepsilon})_{\varepsilon \in (0,\bar{\varepsilon})}$
be a family of normal states such that
\begin{equation}
\label{eq.un}
\forall \alpha \in \N, \, \exists C_{\alpha}>0 \,, \forall \varepsilon \in (0,\bar{\varepsilon}), \; \Tr[\varrho_{\varepsilon}\mathbf{N}^{\alpha}] \leq{C_{\alpha}},
\end{equation}
and
\begin{equation*}
\mathcal{M}(\varrho_{\varepsilon}, \varepsilon \in (0,\bar{\varepsilon})) = \{ \mu \}.
\end{equation*}
Then we say that $(\varrho_{\varepsilon})_{\varepsilon \in
  (0,\bar{\varepsilon}) }$ satisfies the (PI) condition if:\\
\begin{equation}
\label{eq.pi}
\lim_{\varepsilon \to 0} \Tr[\varrho_{\varepsilon}\mathbf{N}^{k}] =
\int_{\mathcal Z}|z|^{2k} d\mu(z),\; \, \forall k \in \N\,.
\end{equation}
\end{dfn}

Let us consider two kinds of normal states on the Fock
space $\Gamma_{s}(\mathcal Z)$, namely the coherent and
Hermite states.
\begin{itemize}
\item
The coherent states are given by
\begin{equation}
\label{eq.coh}
\varrho_{\varepsilon}(f)=|E(f)\rangle \langle E(f)|=| W(\frac{\sqrt{2}}{i\varepsilon}f)\Omega
\rangle \; \langle W(\frac{\sqrt{2}}{i\varepsilon}f)\Omega |\,,
\end{equation}
where $f\in \mathcal Z$  and $\Omega$ is the vacuum vector of the Fock space.
\item
The Hermite states are given by
\begin{equation}
\label{def.hermite}
\varrho_{\varepsilon}(f)=|f^{\otimes k}\rangle \langle
f^{\otimes k}|\,, \quad k=[\frac 1 \varepsilon], \quad f   \in
\mathcal Z.
\end{equation}
\end{itemize}
For some coherent states or Hermite states the (PI) property was proved in $\cite{AmNi2}$ by a
simple computation
\begin{equation}
\lim_{\varepsilon \to 0} \Tr[\varrho_{\varepsilon}\mathbf{N}^{k}]=|f|^{2k}=\delta_{f}(|z|^{2k}).
\end{equation}
Thus there is no loss of compactness for those states when $f$ does
not depend on $\varepsilon$. However, for our examples we will consider
coherent  and Hermite states where the  vector $f$ is
$\varepsilon$-dependent. More precisely, let
$(f_{\varepsilon})_{\varepsilon \in (0,\bar{\varepsilon})}$ be a
family of vectors in $\mathcal Z$ such that : \\

\begin{hyp}
\label{assweak}
 $|f_{\varepsilon}|=1$ and  $f_{\varepsilon}$ weakly convergent to
 $f_{0}$ but not strongly, (i.e.  $|f_{0}|<1$).
\end{hyp}

Then we shall prove that the family of normal states
$\rho_{\varepsilon}(f_{\varepsilon})$ given by \eqref{eq.coh} or \eqref{def.hermite} provides
a good example of states which does not satisfy the (PI) property but
only the uniform condition \eqref{eq.un}.

We recall two useful Propositions from \cite{AmNi1}. Consider two families of vectors
$(u_{\varepsilon}) _{\varepsilon \in
  (0,\bar{\varepsilon})}$ and $(v_{\varepsilon}) _{\varepsilon \in
  (0,\bar{\varepsilon})}$ in $\Gamma_{s}(\mathcal Z)$. With
the family of trace class operators, 
 $$\varrho_{\varepsilon}^{(u_{\varepsilon},v_{\varepsilon})}=|u_{\varepsilon}\rangle
\langle v_{\varepsilon}|,$$
complex-valued Wigner measures can be defined by a simple linear
decomposition, specified in \cite{AmNi1}-Proposition 6.4. Recall that for a family of
Hermite states $u_{\varepsilon}=u^{{\otimes [\frac{1}{\varepsilon}]}}$, $$\mathcal{M}(\varrho_{\varepsilon}^{(u_{\varepsilon},u_{\varepsilon})})=\{\delta^{S^{{1}}}_{u}\}=\{\frac{1}{2\pi}\int_{0}^{2\pi}\delta_{e^{i\theta}u}d\theta\}.$$
\begin{prop}
\label{pr.end2}
Assume that the family  $(u_{\varepsilon})_{\varepsilon \in
  (0,\bar{\varepsilon})}$ and $(v_{\varepsilon})_{\varepsilon \in
  (0,\bar{\varepsilon})}$ of vectors in the Fock space satisfy the uniform estimates
\begin{equation*}
|(1+\mathbf{N})^{\frac{\delta}{2}}u_{\varepsilon}|+|(1+\mathbf{N})^{\frac{\delta}{2}}v_{\varepsilon}|
\leq{C},\; |u_{\varepsilon}|=|v_{\varepsilon}|=1\,,
\end{equation*}
for some $\delta>0$ and $C>0$. If additionally any $\mu \in \mathcal{M}(\varrho_{\varepsilon}^{(u_{\varepsilon},u_{\varepsilon})})$
and any $\nu \in \mathcal{M}(\varrho_{\varepsilon}^{(v_{\varepsilon},v_{\varepsilon})})$ are mutually
orthogonal, then $$\mathcal{M}(\varrho_{\varepsilon}^{(u_{\varepsilon},v_{\varepsilon})},\varepsilon \in
(0,\bar{\varepsilon}))=\{ 0 \}.$$ 
This is equivalent to 
\begin{equation*}
\lim_{\varepsilon \to 0} \langle
u_{\varepsilon},b^{Weyl}v_{\varepsilon}\rangle=0,
\end{equation*}
for any $b \in \mathcal{S}_{cyl}(\mathcal{Z}).$
\end{prop}

\begin{prop}
\label{pr.end}
Assume the same assumptions as in the Proposition above with the additional
condition $\mathcal{M}(\varrho_{\varepsilon}^{(u_{\varepsilon},u_{\varepsilon})})=\{ \mu\}$ and
$\mathcal{M}(\varrho_{\varepsilon}^{(v_{\varepsilon},v_{\varepsilon})})=\{ \nu \}$. Then the
family of trace class operators
$(\varrho_{\varepsilon}^{(u_{\varepsilon}+v_{\varepsilon},u_{\varepsilon}+v_{\varepsilon})})_{\varepsilon
\in (0,\bar{\varepsilon})}$ satisfies
$$\mathcal{M}(\varrho_{\varepsilon}^{(u_{\varepsilon}+v_{\varepsilon},u_{\varepsilon}+v_{\varepsilon})})=
\{\mu+\nu\}.$$
\end{prop}

\begin{cor}
Let $\varrho_{\varepsilon}(f_{\varepsilon})$ be either the family of
coherent states \eqref{eq.coh} or Hermite states \eqref{def.hermite} with
$f_{\varepsilon}$ satisfying {\bf(A3)}. Then the (PI)
condition fails for
$\varrho_{\varepsilon}(f_{\varepsilon})$. However, we get
for all $t>0$, $\mathcal M(\varrho_{\varepsilon}(f_{\varepsilon}) (t),\varepsilon \in (0,\bar
\varepsilon))=\{\mu_{t}\}$
with
$\mu_{t}=\delta_{\Phi(t,0)f_{0}}$ in the case of coherent states and
$\mu_{t}=\delta^{S^{1}}_{\Phi(t,0)f_{0}}$  in the case of Hermite
states, where $\Phi(t,0)$ is the flow associated with
\begin{equation*}
      i  \partial_{t}z_{t}=Az_{t}+\partial_{\bar z}Q(z_{t}).
\end{equation*}
Furthermore, set $u_{\varepsilon}=u^{{\otimes
    [\frac{1}{\varepsilon}]}}$, $u\in \mathcal Z$, then for all $t \in \R$
\begin{equation*}
\mathcal{M}(\varrho_{\varepsilon}^{(u_{\varepsilon}+E(f_{\varepsilon}),
u_{\varepsilon}+E(f_{\varepsilon}))},
\varepsilon \in (0,\bar{\varepsilon}))=\{ \delta_{u_{t}}^{S^{1}}+\delta_{f_{t}} \}\,,
\end{equation*}
with $u_{t}=\Phi(t,0)u$ and $f_{t}=\Phi(t,0)f_{0}$.
\end{cor}

\begin{proof}
The proof splits on several steps.
\begin{itemize}
\item {\it Identification of the Wigner measure:}\\
\underline{For coherent states}:\\
 In order to identify the Wigner measure associated with
  $\varrho_{\varepsilon}(f_{\varepsilon})$, at time $t=0$, use
  the well known formula
\begin{equation*}
\Tr(\varrho_{\varepsilon}(f_{\varepsilon})W(\sqrt{2}\pi\xi))=e^{2i\pi\Re\langle\xi,
f_{\varepsilon}\rangle}e^{-\varepsilon\frac{|\sqrt{2}\pi\xi|^{2}}{4}}\,,
\end{equation*}
here the right hand side converges to $e^{2i\pi\Re\langle\xi,f_0\rangle}=\mathcal
F^{-1}(\delta_{f_{0}})$ when $\varepsilon$ goes to $0$.
Therefore $\mathcal M(\varrho_{\varepsilon}(f_{\varepsilon}),\varepsilon \in (0,\bar
\varepsilon))=\{\delta_{f_{0}}\}$.\\

\underline{For Hermite states}:\\
A simple computation yields for any $b\in
\mathcal{P}^{\infty}_{p,q}(\mathcal Z)$,
\begin{equation*}
\lim_{\varepsilon \to 0}\Tr[\varrho_{\varepsilon}b^{Wick}]= \lim_{\varepsilon \to 0}\langle
f_{\varepsilon}^{\otimes n},b^{Wick}f_{\varepsilon}^{\otimes
  n}\rangle = \langle f_{0}^{\otimes
  q},\tilde{b}f_{0}^{\otimes p}\rangle=\int_{\mathcal Z}b(z)d\delta^{S^1}_{f_{0}}(z).
\end{equation*}
Then by applying Proposition 6.15 in \cite{AmNi1}, one proves that $\mathcal{M}(\varrho_{\varepsilon}(f_\varepsilon),\varepsilon \in
(0,\bar{\varepsilon}))=\{ \delta^{S^1}_{f_{0}} \}$.

\underline{For the superposition of orthogonal states}:\\
Recall
that $\varrho_{\varepsilon}^{(u+E(f_{\varepsilon}),u+E(f_{\varepsilon}))}=|u+E(f_{\varepsilon})\rangle
\langle u+E(f_{\varepsilon})|$ and according to Proposition
\ref{pr.end} and \ref{pr.end2}
\begin{align*}
\mathcal{M}(\varrho_{\varepsilon}^{(u_{\varepsilon}+E(f_{\varepsilon}),u_{\varepsilon}+E(f_{\varepsilon}))},
\varepsilon \in (0,\bar{\varepsilon}))&=\mathcal{M}(\varrho_{\varepsilon}^{(u_{\varepsilon},u_{\varepsilon})})+\mathcal{M}(\varrho_{\varepsilon}^{(E(f_{\varepsilon}),E(f_{\varepsilon}))})+\{0\}&\\
&=\{\delta_{u}^{S^{1}}+\delta_{f_{0}}\}.&
\end{align*}

\item {\it Uniform estimates:}\\
Let  $k \in \mathbb N$, the following uniform estimate holds\\
\underline{For the coherent states}:\\
\begin{equation*}
\Tr(\varrho_{\varepsilon}(f_{\varepsilon})\mathbf{N}^{k})=\langle \Omega,
W^{*}(\frac{\sqrt{2}}{i\varepsilon}f_{\varepsilon})\mathbf{N}^{k}W(\frac{\sqrt{2}}
{i\varepsilon}f_{\varepsilon})\Omega
\rangle \leq C_{k}|(\mathbf{N}+1)^{k/2}\Omega|^{2} \leq C_{k}.
\end{equation*}\\
\underline{For Hermite states}:\\
In this case $\varrho_{\varepsilon}(f_{\varepsilon})=|f_{\varepsilon}^{\otimes N}\rangle \langle f_{\varepsilon}^{\otimes N}|$
with $N=[\frac{1}{\varepsilon} ]$ \; \text{is the number of
  particles.}\\ Notice that for all $k \in \N$
\begin{equation}
\label{eq.good}
\Tr[\varrho_{\varepsilon}\mathbf{N}^{k}]=(\varepsilon
n)^{k}|f_{\varepsilon}|^{2}=(\varepsilon n)^{k}.
\end{equation}
\underline{For the superposition of orthogonal states}:\\
In this case
$\varrho_{\varepsilon}^{(u_{\varepsilon}+E(f_{\varepsilon}),u_{\varepsilon}+E(f_{\varepsilon}))}=|u_{\varepsilon}+E(f_{\varepsilon})\rangle
\langle u_{\varepsilon}+E(f_{\varepsilon})|$, 
\begin{align*}
\forall k \in \N,\;\Tr[(|u_{\varepsilon}+E(f_{\varepsilon})\rangle
\langle u_{\varepsilon}+E(f_{\varepsilon})|)\mathbf{N}^{k}] \leq C_{k}.
\end{align*}

\item {\it The condition (PI) fails:}\\
\underline{For coherent state}s:\\
 A simple computation of
  $\Tr(\varrho_{\varepsilon}(f_{\varepsilon})\mathbf{N}^{{k}})$ when $k=1$ gives
  the following equality
$$\Tr(\varrho_{\varepsilon}(f_{\varepsilon})\mathbf{N})=\langle E(f_{\varepsilon}),(|z|^{2})^{Wick}
E(f_{\varepsilon})\rangle=|f_{\varepsilon}|_{\mathcal Z}^{2}=1,$$
and
$$\int_{\mathcal Z}|z|^{2}d\delta_{f_{0}}(z)=|f_{0}|_{\mathcal Z}^{2}\,.$$
But $|f_{\varepsilon}|_{\mathcal Z}^{2}$ does not converge to
$|f_0|_{\mathcal Z}^{2}$ because $f_{\varepsilon}$ does not
converge strongly to $f_{0}$. Hence the quantity
$\Tr(\varrho_{\varepsilon}(f_{\varepsilon})\mathbf{N})$ does not converge to $\int_{\mathcal Z}|z|^{2}d\delta_{f_{0}}(z)$.
Then the (PI) condition is not satisfied.

\underline{For the Hermite states}:\\
It is easy to see that the (PI) condition fails. Indeed on the one hand
$\lim_{\varepsilon \to 0}\Tr[\varrho_{\varepsilon}\mathbf{N}^{k}]=1$ but on the
other hand $\int_{\mathcal Z}|z|^{2k}d\delta_{f_{0}}=|f_{0}|^{2k}<1.$\\

\underline{For the superposition of orthogonal states}:\\
Assume that the family
$(\varrho_{\varepsilon}^{(u_{\varepsilon}+E(f_{\varepsilon}),u_{\varepsilon}+E(f_{\varepsilon}))})_{\varepsilon
\in (0,\bar{\varepsilon})}$ satisfies the (PI) condition. Fix
$k=1$ and compute:
\begin{align*}
\Tr[\varrho_{\varepsilon}^{(u_{\varepsilon}+E(f_{\varepsilon}),u_{\varepsilon}+E(f_{\varepsilon}))}\mathbf{N}^{k}]
&=&\langle
u_{\varepsilon},\mathbf{N}u_{\varepsilon}\rangle + \langle
E(f_{\varepsilon}),\mathbf{N}E(f_{\varepsilon})\rangle + \langle
u_{\varepsilon},\mathbf{N}E(f_{\varepsilon})\rangle \\&&+ \langle
E(f_{\varepsilon}),\mathbf{N}u_{\varepsilon}\rangle
\end{align*}
The two last terms converge to $0$ when $\varepsilon \to 0$ since 
\begin{align*}
\lim_{\varepsilon \to 0, \, \varepsilon n \to 1} |\langle
E(f_{\varepsilon}),\mathbf{N}u_{\varepsilon}\rangle|&=\lim_{\varepsilon
  \to 0, \, \varepsilon n \to 1} |\langle
u_{\varepsilon},\mathbf{N}E(f_{\varepsilon})\rangle| \\
&= \lim_{\varepsilon \to 0, \, \varepsilon n \to 1} \varepsilon n|\langle u_{\varepsilon},\frac{\varepsilon^{-\frac{n}{2}}e^{-\frac{1}{2\varepsilon}}}{\sqrt{n!}}f_{\varepsilon}\rangle|=0,&
\end{align*}
according to formula $$E(f_{\varepsilon})=\sum_{n=0}^{\infty}\frac{\varepsilon^{-\frac{n}{2}}e^{-\frac{|f_{\varepsilon}|^{2}}{2\varepsilon}}}{\sqrt{n!}}f_{\varepsilon}^{\otimes
n}.$$
Besides, the family of Hermite states $(\varrho_{\varepsilon}^{(u_{\varepsilon},u_{\varepsilon})})_{\varepsilon
\in (0,\bar{\varepsilon})}$ satisfies the (PI) condition for $k=1$, hence 
the family $(\varrho_{\varepsilon}^{(E(f_{\varepsilon}),E(f_{\varepsilon}))})_{\varepsilon
\in (0,\bar{\varepsilon})}$ satisfies the (PI) condition for $k=1$ which is wrong.

\item {\it Propagation}:\\
All the examples (coherent, Hermite and orthogonal states) satisfy the assumptions of Theorem \ref{thm.main}. Hence, for all $t>0$
 $$\mathcal M(\varrho_{\varepsilon}(f_{\varepsilon})(t),\varepsilon \in (0,\bar
\varepsilon))=\{\mu_{t}\},$$ with  
$\mu_{t}=\Phi(t,0)_*\mu_0$ and where $\mu_0$ is the initial Wigner
measure of coherent, Hermite and orthogonal states previously computed. 
\end{itemize}
\end{proof}

\appendix

\section{Second quantization}
\subsection{Fock space}
Let $\mathcal Z$ be a Hilbert space with the inner product $\langle.\,,.\rangle$
antilinear on the left-hand side associated with a norm
$|z|=\sqrt{\langle z\,,z\rangle}$.
\begin{dfn}
 The bosonic Fock space on $\mathcal Z$ is
 given by:
\begin{equation*}
\Gamma_{s}(\mathcal Z)=\bigoplus_{n=0}^{\infty}\bigvee^{n}\mathcal
Z,
\end{equation*}
where $\bigvee^{n}\mathcal Z$ denotes the n-fold symmetric tensor
product. For all n $\in \N$ the orthogonal projection of
$\bigotimes^{n}\mathcal Z$ on the subspace $\bigvee^{n}\mathcal
Z$ is denoted by $\mathcal{S}_{n}$. Moreover $\mathcal{S}_{n}$ have the explicit writing, for all $\xi_{1},...,\xi_{n} \in \mathcal Z$:
\begin{equation*}
\xi_{1} \vee \xi_{2} \vee ... \vee \xi_{n}= S_{n}(\xi_{1}\otimes
\xi_{2} \otimes ... \otimes \xi_{n})=\frac{1}{n!}\sum_{\sigma \in
  \Sigma_{n}}\xi_{\sigma(1)} \otimes \xi_{\sigma(2)} \otimes ... \otimes
\xi_{\sigma(n)},
\end{equation*}
where $\Sigma_{n}$ is the n-th fold symmetric group.
\end{dfn}
The algebraic direct sum is $\Gamma_{s}^{fin}(\mathcal Z)=\bigoplus_{n=0}^{alg}\bigvee^{n}\mathcal Z$.
\begin{prop}
The family $(\xi_{1} \vee \xi_{2} \vee ... \vee \xi_{n})_{\xi_{i} \in
  \mathcal Z, i=1,...,n}$ spans $\bigvee^{n,alg}\mathcal Z$ and is a
total family of $\bigvee^{n}\mathcal Z$. The same property holds
for $(z^{\otimes n})_{z \in \mathcal Z,n \in \N}$.
\end{prop}
\begin{proof}
$\mathcal{S}_{n}$ is an orthogonal projection since
$\bigvee^{n}\mathcal Z$ is a closed subspace, then the family
$(\xi_{1} \vee \xi_{2} \vee ... \vee \xi_{n})_{\xi_{i} \in \mathcal Z}$ spans $\bigvee^{n,alg}\mathcal Z$ and is a total family of $\bigvee^{n}\mathcal Z$.
The last result is straightforward from the equality:
\begin{equation*}
\xi_{1} \vee \xi_{2} \vee ... \vee \xi_{n} = \frac{1}{2^{n}n!}\sum_{\varepsilon_{i}=\pm 1}\varepsilon_{1}...\varepsilon_{n}\big(\sum_{j=1}^{n}\varepsilon_{j}\xi_{j}\big)^{\otimes n}.
\end{equation*}
\end{proof}

\subsection{Main operators}
\label{se.main}
For $k=1,2$ and any operators $A_{k}: \bigvee^{i_{k}}\mathcal Z
\rightarrow \bigvee^{j_{k}}\mathcal Z$ we can define the symmetric
tensor product of operators:

\begin{equation*}
A_{1}\bigvee A_{2}= S_{j_{1}+j_{2}} \circ (A_{1}\otimes A_{2}) \circ S_{i_{1}+i_{2}} \in \mathcal L(\bigvee^{i_{1}+i_{2}} \mathcal Z, \bigvee^{j_{1}+j_{2}}\mathcal Z).
\end{equation*}

For all $z \in \mathcal Z$, let denote $|z\rangle$ the operator: $\lambda
\in \C \mapsto \lambda z \in \mathcal Z $ and $\langle z|$ the
linear functional: $\xi \mapsto \langle z,\xi\rangle \in \C$. Now let introduce the annihilation and creation operators.
\begin{dfn}
For $z \in \mathcal Z$, $n \in \N^{*}$ and $\varepsilon>0$ a
parameter. The $\varepsilon$-dependent annihilation and creation operators are defined by:
\begin{align*}
a(z)_{|\vee^{n+1}\mathcal Z}&= \sqrt{\varepsilon (n+1)} \langle z| \otimes \Id_{\vee^{n}\mathcal Z},&\\
a^{*}(z)_{|\vee^{n}\mathcal Z}&=\sqrt{\varepsilon (n+1)} S_{n+1} \circ
( |z\rangle \otimes \Id_{\vee^{n}\mathcal Z})=\sqrt{\varepsilon (n+1)}
\;|z \rangle \bigvee \Id_{\vee^{n}\mathcal Z}&.
\end{align*}
\end{dfn}
The families $(a(z))_{z \in \mathcal Z}$ and $(a^{*}(z))_{z \in
  \mathcal Z}$  satisfy the canonical commutation relations for all $z_{1},z_{2} \in \mathcal Z$:
\begin{equation*}
[a(z_{1}),a^{*}(z_{2})]=\varepsilon \langle z_{1},z_{2}\rangle \Id, \quad 
[a(z_{1}),a(z_{2})]=0, \quad [a^{*}(z_{1}),a^{*}(z_{2})]=0. 
\end{equation*}
We also consider another important operator, namely the field
operator $$\Phi(z)=\frac{1}{\sqrt{2}}(a^{*}(z)+a(z))$$ generator of
the unitary group $W(z)=e^{i\Phi(z)}$ which satisfies the Weyl
commutation relations for all $z_{1}$,$z_{2}$ $\in \mathcal Z$,
\begin{equation*}
W(z_{1})W(z_{2})=e^{-\frac{i\varepsilon}{2}\Im \langle z_{1},z_{2} \rangle}W(z_{1}+z_{2})\;.
\end{equation*}
The number operator $\mathbf N$, parametrized by $\varepsilon>0$, is
defined according to
\begin{equation*}
\mathbf{N}_{|\vee^{n}\mathcal Z}=\varepsilon n \Id_{\vee^{n}\mathcal Z}.
\end{equation*}
Finally, we recall that $\mathrm{d}\Gamma(A)$ is given by:
\begin{equation*}
\mathrm{d}\Gamma(A)_{|\vee^{n,alg}D(A)}=
\varepsilon  \sum_{k=1}^{n}\Id^{\otimes(k-1)}\otimes A \otimes \Id^{\otimes(n-k)}
\end{equation*}
In particular, $\mathbf{N}=\mathrm{d}\Gamma(\Id)$.

\subsection{Wick quantization}
\label{se.wickq}
For all p,q $\in \N$, we denote $\mathcal P_{p,q}(\mathcal Z)$ the space of complex-valued polynomials on $\mathcal Z$, defined by the following continuity condition
\begin{equation}
b \in \mathcal P_{p,q}(\mathcal Z) \Leftrightarrow \exists \tilde{b} \in
\mathcal{L}(\bigvee^{p}\mathcal Z,\bigvee^{q}\mathcal Z),
b(z)=\langle z^{\otimes q},\tilde{b}z^{\otimes p}\rangle.
\end{equation}
These spaces are equipped with norms $|.|_{\mathcal P_{p,q}}$: $$|b|_{\mathcal P_{p,q}}=\|\tilde{b}\|_{\mathcal{L}(\bigvee^{p}\mathcal Z\bigvee^{q}\mathcal Z)}.$$ The subspace of $\mathcal P_{p,q}(\mathcal Z)$ polynomials b such that $\tilde{b}$ is a compact operator is denoted by $\mathcal P^{\infty}_{p,q}(\mathcal Z)$.
The Wick quantization corresponds to each symbol $b(z) \in \mathcal
P_{p,q}(\mathcal Z)$ a linear operator $b^{Wick}:$  $\Gamma_{s}^{fin}(\mathcal Z)$ $\longrightarrow$ $\Gamma_{s}^{fin}(\mathcal Z)$.
\begin{dfn}
For each symbol $b(z) \in \mathcal P_{p,q}(\mathcal Z)$, is associated
an operator:
$\Gamma_{s}^{fin}(\mathcal Z) \longrightarrow
\Gamma_{s}^{fin}(\mathcal Z)$, given by
\begin{equation*}
b_{|\bigvee^{n}\mathcal Z}^{Wick}=1_{[p,+\infty)}(n) \frac{\sqrt{n!(n+q-p)!}}{(n-p)!}\varepsilon^{\frac{p+q}{2}}S_{n-p+q}(\tilde{b}\otimes \Id^{\otimes (n-p)})\;.
\end{equation*}
\end{dfn}
We recall some well-known number estimate:
\begin{prop}
For $b \in \mathcal P_{p,q}(\mathcal Z)$, the following estimate:
\begin{equation*}
\|\langle\mathbf{N}\rangle^{-\frac{q}{2}}b^{Wick}\langle\mathbf{N}\rangle^{\frac{p}{2}}\|_{\mathcal
  L(\mathcal Z)} \leq
|b|_{\mathcal{P}_{p,q}}.
\end{equation*}
holds with $\langle\mathbf{N}\rangle=(1+\mathbf{N}^{2})^{\frac{1}{2}}$.
\end{prop}
\begin{prop}
\label{number}
Let $b \in \mathcal P_{p,q}(\mathcal Z)$,
$\langle\mathbf{N}\rangle^{-\frac{p+q}{2}}b^{Wick}$ and 
$b^{Wick}\langle\mathbf{N}\rangle^{-\frac{p+q}{2}}$ extend to bounded
operators on $\mathcal Z$ with norms smaller than
$C_{p,q}|b|_{\mathcal{P}_{p,q}}$.
\end{prop}
An important operation with the Wick symbols is  the composition:
$b_{1}^{Wick} \circ \; b_{2}^{Wick}$ with $b_{1}$,$b_{2} \in
\bigoplus_{p,q}^{alg} \mathcal P_{p,q}(\mathcal Z)$ which is a Wick symbol
in $\bigoplus_{p,q}^{alg}\mathcal P_{p,q}(\mathcal Z)$. Now we introduce the useful notations for the formula about the composition.\\
Let $b \in \mathcal P_{p,q}(\mathcal Z)$, the k-th differential of  b is well defined and
\begin{equation*}
\partial_{z}^{k}b(z) \in (\bigvee^{k}\mathcal Z)^{*},\quad \mbox{ and } \quad 
\partial_{\bar z}^{k}b(z) \in \bigvee^{k}\mathcal Z.
\end{equation*}
Furthermore, we can define any differential
$\partial_{\bar{z}}^{{j}}\partial_{z}^{{k}}b(z)$ for $j,k \in \N$ at
the point $z \in \mathcal{Z}$:
\begin{equation}
\partial_{\bar{z}}^{{j}}\partial_{z}^{{k}}b(z)=\frac{p!}{(p-k)!}\frac{q!}{(q-j)!}
(\langle z^{{\otimes q-j}}| \bigvee \Id_{\vee^{{j}}\mathcal{Z}})
\tilde{b} (|z^{\otimes p-k}\rangle\bigvee \Id_{\vee^{k}\mathcal{Z}}) \in \mathcal{L}
(\vee^{{k}}\mathcal{Z},\vee^{{j}}\mathcal{Z}).
\end{equation}
We use the following notations about the Poisson brackets:
\begin{equation*}
\{ b_{1},b_{2} \} ^{(k)}(z) = \partial_{z}^{k}b_{1}(z).\partial_{\bar z}^{k}b_{2}(z)-\partial_{z}^{k}b_{2}(z).\partial_{\bar z}^{k}b_{1}(z).
\end{equation*}
with the $\C-$bilinear duality product
$\partial_{z}^{k}b_{1}(z).\partial_{\bar
  z}^{k}b_{2}(z)=<\partial_{z}^{k}b_{1}(z),\partial_{\bar
  z}^{k}b_{2}(z)>_{((\bigvee^{k}\mathcal
  Z)^{*},\bigvee^{k}\mathcal Z)}$,\\which defines a function of
$z \in \mathcal Z$ simply denoted by $\partial_{z}^{k}b_{1}.\partial_{\bar
  z}^{k}b_{2}.$
\begin{prop}
Let $b_{1} \in \mathcal P_{p_{1},q_{1}}(\mathcal Z)$ et $b_{2} \in \mathcal P_{p_{2},q_{2}}(\mathcal Z)$. For all $k \in \{ 0,..., \min{(p_{1},q_{2})} \}$, $\partial_{z}^{k}b_{1}.\partial_{\bar z}^{k}b_{2}$ belongs to $\mathcal P_{p_{1}+p_{2}-k,q_{1}+q_{2}-k}(\mathcal Z)$, we have the estimate:
\begin{equation*}
|\partial_{z}^{k}b_{1}.\partial_{\bar z}^{k}b_{2}|_{\mathcal P_{p_{1}+p_{2},q_{1}+q_{2}}} \leq \frac{p_{1}!}{(p_{1}-k)!}\frac{q_{2}!}{(q_{2}-k)!}|b_{1}|_{\mathcal P_{p_{1},q_{1}}}|b_{2}|_{\mathcal P_{p_{2},q_{2}}},
\end{equation*}
and the following formulas hold true on $\Gamma_{s}^{fin}(\mathcal Z)$:
\begin{itemize}
\item
\begin{equation*}
b_{1}^{Wick} \circ b_{2}^{Wick} =\Big[ \sum_{k=0}^{\min(p_{1},q_{2})}\frac{\varepsilon^{k}}{k!} \partial_{z}^{k}b_{1}.\partial_{\bar z}^{k}b_{2} \Big]^{Wick}.
\end{equation*}
\item
\begin{equation*}
[b_{1}^{Wick},b_{2}^{Wick}]=\sum_{k=1}^{\max(\min(p_{1},q_{2}),\min{(p_{2},q_{1})})}\frac{\varepsilon^{k}}{k!} \Big[\{b_{1},b_{2} \}^{(k)} \Big]^{Wick}.
\end{equation*}
\end{itemize}
\end{prop}

The following lemma specifies the behaviour of the Wick observables
conjugated by the Weyl operators and unitary groups generated by free Hamiltonians (see \cite{AmNi3} for 
more details).
\begin{lm}
Let $b \in \mathcal{P}_{alg}(\mathcal Z)$:\\
\noindent\textbf{a)}  The operator $b^{Wick}$ is closable and the domain of his closure contains:
\begin{equation*}
\mathcal H_{0}=Span \{ W(\phi)\psi, \psi \in \Gamma_{s}^{fin}(\mathcal Z),\phi \in \mathcal Z \}.
\end{equation*}
\noindent\textbf{b)}  For all $\xi \in \mathcal Z$, the equality
\begin{equation}
\label{eq.toto11}
W(\sqrt{2}\pi\xi)^{*}b^{Wick}W(\sqrt{2}\pi\xi)=\{ b ( z + i\pi\varepsilon\xi ) \}^{Wick}.
\end{equation}
holds on $\mathcal H_{0}.$\\
\noindent\textbf{c)} Let A be a self-adjoint operator on $\mathcal Z$ then for all $t\in\mathbb{R}$,
\begin{equation}
\label{eq.toto7}
e^{i\frac{t}{\varepsilon}d\Gamma(A)}b^{Wick}e^{-i\frac{t}{\varepsilon}d\Gamma(A)}=
(b(e^{-itA}z))^{Wick}.
\end{equation}
\end{lm}

\begin{lm}
Any $b(z) \in \bigoplus_{j=0}^{r}\mathcal{P}_{j,j}(\mathcal Z)$
satisfies the following properties:
\begin{enumerate}
\item The equality
\begin{equation}
\label{eq.toto10}
b(z+i\pi\varepsilon \xi)=\sum_{j=0}^{r}\frac{(i\varepsilon\pi)^{j}}{j!}\mathrm{D}^{j}[b(z)][\xi],
\end{equation}
where $\mathrm{D}^{j}[b(z)][\xi]$ the j-th differential of $b$ with respect to
$(z,\bar{z})$ evaluated at $\xi$, i.e:
\begin{equation*}
\mathrm{D}^{j}[b(z)][\xi]=\sum_{|\alpha|+|\beta|=j}\frac{j!}{\alpha!\beta!}\langle
\xi^{\otimes \beta}\; , \;\partial_{z}^{\alpha}\partial_{\bar
z}^{\beta}b(z){\xi}^{\otimes \alpha}\rangle.
\end{equation*}
\item
There exists a $\varepsilon$-independent constant $C_{r}>0$ such that
\begin{equation*}
\|\langle\mathbf{N}\rangle^{-\frac{r}{2}}\sum_{j=0}^{r}\frac{(i\varepsilon
  \pi)^{j}}{j!}(\mathrm{D}^{j}[b(z)][\xi])^{Wick}\langle\mathbf{N}\rangle^{-\frac{r}{2}}\|_{\mathcal
  L(\Gamma_{s}(\mathcal Z))} \leq C_{r}\langle\xi\rangle^{r}.
\end{equation*}
\end{enumerate}
\end{lm}
\begin{proof}
The first statement follows by Taylor expansion. Then notice that for all $j \in
[0,r]$, $\mathrm{D}^{j}[b(z)][\xi] \in \bigoplus_{m,n}^{r-j}\mathcal{P}_{m,n}(\mathcal
Z)$. The number estimate of Proposition \ref{number} implies the existence of $C_{r}>0$ such that
\begin{equation}
 \|\langle\mathbf{N}\rangle^{-\frac{r}{2}}\sum_{j=0}^{r}\frac{(i\varepsilon
  \pi)^{j}}{j!}(\mathrm{D}^{j}[b(z)][\xi]) ^{Wick}\langle\mathbf{N}\rangle^{-\frac{r}{2}}\|_{\mathcal
  L(\Gamma_{s}(\mathcal Z))} \leq C_{r}\langle\xi\rangle^{r}.
\end{equation}
\end{proof}

\section{Wigner measures}
\label{se.Wigdef}
The Wigner measures are defined in \cite[theorem 6.4]{AmNi1}. We
recall here the main result.
\begin{thm}
\label{th.Wigdef}
Let $(\varrho_{\varepsilon})_{\varepsilon \in (0,\bar{\varepsilon})}$ be
a family of normal states on $\Gamma_{s}(\mathcal Z)$. Assume that
$\Tr[\varrho_{\varepsilon}\mathbf{N}^{\delta}] \leq{C_{\delta}}< +\infty$ uniformly with respect to $\varepsilon \in (0,\bar{\varepsilon})$ for some $\delta>0$ fixed and $C_{\delta} \in \R_{+}^{*}$. Then for any sequence $(\varepsilon_{n})_{n \in \N}$ with $\lim_{n \to +\infty}\varepsilon_{n}=0$, there is an extracted subsequence $(\varepsilon_{n_{k}})_{k \in \N}$ and a Borel probability measure on $\mathcal Z$ such that:
\begin{equation*}
\lim_{k \to +\infty}\Tr[\varrho_{\varepsilon_{n_{k}}}b^{Weyl}]=\int_{\mathcal Z}b(z) d\mu(z),
\end{equation*}
for all $b \in S_{cyl}(\mathcal Z)$. Moreover the probability measure $\mu$ fulfills
\begin{equation*}
\int_{\mathcal Z} |z|^{2\delta}d\mu(z) <C_\delta< \infty.
\end{equation*}
\end{thm}

\begin{dfn}
The set of Wigner measures associated with a family $(\varrho_{\varepsilon})_{\varepsilon \in (0,\bar{\varepsilon})} $ ( respectively a sequence $(\varrho_{\varepsilon_{n}})_{n \in \N})$ which follows the hypotheses of the previous theorem is denoted:
$\mathcal{M}(\varrho_{\varepsilon},\varepsilon \in (0,\bar{\varepsilon}))$, (respectively $\mathcal{M}(\varrho_{\varepsilon_{n},n \in \N}))$
The expression $\mathcal M(\varrho_{\varepsilon},\varepsilon \in
(0,\bar{\varepsilon}))=\{ \mu \}$ means that the family $(\varrho_{\varepsilon})_{\varepsilon \in (0,\bar{\varepsilon})}$ is pure in the sense :
\begin{equation*}
\lim_{\varepsilon \to 0}\Tr[\varrho_{\varepsilon}b^{Weyl}]=\int_{\mathcal Z}b(z) d\mu(z),
\end{equation*}
for all cylindrical symbols b $\in \mathcal S_{cyl}(\mathcal Z)$
without extracting a subsequence.
\end{dfn}
We can assume without loss of generality that $\mathcal M(\varrho_{\varepsilon},\varepsilon
\in (0,\bar{\varepsilon}))=\{ \mu \}$ to prove properties of $\mathcal
M(\varrho_{\varepsilon},\varepsilon \in (0,\bar{\varepsilon}))$.
\\
In practice the Wigner measures  are identified though their characteristic functions with the relation:
\begin{align*}
\mathcal M(\varrho_{\varepsilon},\varepsilon \in (0,\bar{\varepsilon}))&=\{ \mu \} \Leftrightarrow \lim_{\varepsilon \to 0}\Tr[\varrho_{\varepsilon}W(\sqrt{2}\pi\xi)]=\mathcal{F}^{-1}(\mu)(\xi),&\\
&\Leftrightarrow \lim_{\varepsilon \to 0}\Tr[\varrho_{\varepsilon}W(\xi)]=\int_{\mathcal Z}e^{i\sqrt{2}\Real \langle\xi,z \rangle}d\mu(z).
\end{align*}
An a
priori estimate argument allows to extend the previous definition to Wick symbols with compact kernels :

\begin{lm} \label{eq.wigcompact}
Let $(\varrho_{\varepsilon})_{\varepsilon \in (0,\bar{\varepsilon})}$ be
a family of normal states on $\mathcal{L}(\mathcal Z)$ depending on $\varepsilon$ such that
\begin{equation*}
\forall \alpha \in \N, \exists C_\alpha>0, \,\forall \varepsilon\in(0,\bar\varepsilon), \quad \Tr[\varrho_{\varepsilon}\mathbf{N}^{\alpha}]\leq{C_{\alpha}},
\end{equation*}
and $\mathcal M(\varrho_{\varepsilon},\varepsilon
\in (0,\bar{\varepsilon}))=\{ \mu \}$. Then, for any $ b \in \mathcal{P}_{alg}^{\infty}(\mathcal{Z})$,
\begin{equation*}
\lim\limits_{\varepsilon \to 0}
\Tr[\varrho_{\varepsilon}b^{Wick}]=\int_{\mathcal{Z}}b(z)d\mu(z)\;.
\end{equation*}
\end{lm}

\section{Results in infinite dimension for a transport equation}
Recall that the Wasserstein distance is given by the formula
\begin{equation}
W_{2}(\mu_{1},\mu_{2})=\sqrt{\inf_{\mu \in \Gamma(\mu_{1},\mu_{2})}
\int_{\mathcal Z} | z_{1} - z_{2} |_{\mathcal Z}^{2} d\mu(z_{1},z_{2})},
\end{equation}
with $\Gamma(\mu_{1},\mu_{2})$ is the set of probability measures $\mu$ on
$\mathcal Z\times \mathcal Z$ such that the marginals
$(\Pi_{1})_{*}\mu=\mu_{1}$ and $(\Pi_{2})_{*}\mu=\mu_{2}$.
The following result is the second part of Theorem 8.3.1 in \cite{AGS}
with p=2. 

\begin{prop}
\label{prop.C1}
Let $I$ be an open interval in $\R$. If a weakly narrowly continuous curve
$\mu_{t}$: $I\to Prob_{2}(\mathcal Z)$ satisfies the continuity equation
\begin{equation*}
\partial_{t}\mu_{t}+\nabla^{T}(v_{t}\mu_{t})=0,
\end{equation*}
in the weak sense:
\begin{equation*}
  \int_{\R}\int_{\mathcal
    Z}(\partial_{t}\phi(z,t)+<v_{t}(z),\nabla_{z}\phi(z,t)>_{\mathcal
    Z})d\mu_{t}(z)dt=0, \forall \phi \in
  \mathcal{C}_{0,cyl}^{{\infty}} (\R \times \mathcal Z),
\end{equation*}
for some Borel velocity field $v_{t}$, with
$|v_{t}(z)|_{L^{2}(\mathcal Z,\mu_{t})} \in L^{1}(I)$, then $\mu_{t}$
is absolutely continuous with $W_{2}(\mu_{t},\mu_{t'}) \leq
\int_{t}^{t'}|v_{s}|_{L^{2}(\mathcal Z,\mu_{s})} ds$. Moreover for
Lebesgue almost every $t \in I$, $v_{t}$ belongs to the closure in
$L^{2}(\mathcal Z,\mu_{t})$ of the subspace spanned by $\{\nabla
\varphi, \varphi \in \mathcal C_{0,cyl}(\mathcal Z) \}.$
\end{prop}
\begin{prop}
\label{prop.C4}
Let $\mu_{t}: \R \mapsto Prob_{2}(\mathcal Z)$ be a $W_{2}$-continuous solution to the equation:
\begin{equation*}
\partial_{t}\mu_{t}+\nabla^{T}(v_{t}\mu_{t})=0,
\end{equation*}
in the weak sense:
\begin{equation*}
  \int_{\R}\int_{\mathcal Z}(\partial_{t}\phi(z,t)+\langle
  v_{t}(z),\nabla_{z}\phi(z,t)\rangle_{\mathcal Z})d\mu_{t}(z)dt=0,
  \forall \phi \in \mathcal{C}_{0,cyl}^{\infty}(\R \times \mathcal{Z}),
\end{equation*}
for a suitable $v(t,z)=v_{t}(z)$ such that $|v_{t}(z)|_{L^{2}(\mathcal Z,\mu_{t})} \in L^{1}([-T,T])$ for all $T>0$. Assume additionally that the Cauchy problem
\begin{equation*}
\partial_{t}\gamma(t)=v_{t}(\gamma(t)), \gamma(s)=x,
\end{equation*}
admits a unique global continuous solution on $\R$ for all $s \in \R$, and for all $z \in \mathcal Z$ such that $\gamma(t,s)=\phi(t,s)\gamma(s)$ defines a Borel flow on $\mathcal Z$. Then the measure $\mu_{t}$ satisfies $$\mu_{t}=\phi(t,s)_{*}\mu_{s}.$$
\end{prop}

\bibliographystyle{plain}
\bibliography{article1}

\end{document}